\documentclass[11pt,a4paper]{article}
 \pdfoutput=1
 \usepackage{amsfonts,amsmath,amscd,amssymb,amsbsy,amsthm,amstext,amsopn,amsxtra,comment}
\usepackage{natbib}
\usepackage{fullpage,mathrsfs,subfigure}
\usepackage[british]{babel} 
\usepackage{pstricks,pst-node,pst-text,pst-tree,multicol}
 \usepackage{stmaryrd}
 \usepackage{extarrows}
\usepackage[all]{xy}
\usepackage{graphicx}
\usepackage{enumitem}
\usepackage{geometry}
\usepackage[margin=10pt,font=small,labelfont=bf,labelsep=endash]{caption}
\usepackage{pdfpages}
\usepackage{hyperref}
\usepackage{xcolor}
\usepackage{pdflscape}

\newcommand{\BBs}{\ensuremath{\mathbf{B}_{\mathrm{sp}}}}
\newcommand{\BBst}{\ensuremath{\widetilde{\mathbf{B}}_{\mathrm{sp}}}}
\newcommand{\IB}{\ensuremath{\mathbf{I}}}

\newcommand{\WB}{\ensuremath{\mathbf{W}}}

\newcommand{\rb}{\ensuremath{\mathbf{r}}}
\newcommand{\vb}{\ensuremath{\mathbf{v}}}
\newcommand{\fb}{\ensuremath{\mathbf{f}}}
\newcommand{\gb}{\ensuremath{\mathbf{g}}}
\newcommand{\cb}{\ensuremath{\mathbf{c}}}
\newcommand{\bb}{\ensuremath{\mathbf{b}}}
\newcommand{\xb}{\ensuremath{\mathbf{x}}} 
\newcommand{\yb}{\ensuremath{\mathbf{y}}} 
\newcommand{\zb}{\ensuremath{\mathbf{z}}} 
\newcommand{\tb}{\ensuremath{\mathbf{t}}}
\newcommand{\PPhi}{\ensuremath{\boldsymbol{\Phi}}}
\newcommand{\SBl}{\ensuremath{\mathbf{S}_{\lambda}}}
\newcommand{\SBlx}{\ensuremath{\mathbf{S}_{\lambda_x}}}
\newcommand{\betah}{\ensuremath{\hat{\beta}}}
\newcommand{\betahps}{\ensuremath{\hat{\beta}}}
\newcommand{\betahpr}{\ensuremath{\hat{\beta}_{\mathrm{pr}}}}
\newcommand{\betahts}{\ensuremath{\hat{\beta}^+}}
\newcommand{\fbh}{\ensuremath{\hat{\mathbf{f}}}}
\newcommand{\fbhps}{\ensuremath{\mathbf{\hat{f}}}}
\newcommand{\fbhpr}{\ensuremath{\mathbf{\hat{f}}_{\mathrm{pr}}}}
\newcommand{\fbhts}{\ensuremath{\mathbf{\hat{f}}^+}}
\newcommand{\fbhpsi}{\ensuremath{\hat{f}_i}}
\newcommand{\fbhtsi}{\ensuremath{\hat{f}^+_i}}
\newcommand{\eepsilon}{\ensuremath{\boldsymbol{\epsilon}}}
\newcommand{\pphi}{\ensuremath{\boldsymbol{\phi}}}
\newcommand{\muh}{\ensuremath{\hat{\mu}}}
\newcommand{\etah}{\ensuremath{\hat{\eta}}}

\newcommand{\LO}{\ensuremath{\mathcal{O}}}
\newcommand{\lo}{\ensuremath{o}}

\newcommand{\xxi}{\ensuremath{\boldsymbol{\xi}}}

\newcommand{\E}{\ensuremath{\mathrm{E}}}
\newcommand{\var}{\ensuremath{\mathrm{Var}}}

\newcommand{\amse}{\ensuremath{\mathrm{AMSE}}}
\newcommand{\Btp}{\ensuremath{B_{\mathrm{tp}}^2}}
\newcommand{\Bps}{\ensuremath{B^2}}
\newcommand{\Bts}{\ensuremath{B_+^2}}
\newcommand{\Bpr}{\ensuremath{B_{\mathrm{pr}}^2}}
\newcommand{\Vps}{\ensuremath{V}}
\newcommand{\Vts}{\ensuremath{V_+}}
\newcommand{\Vpr}{\ensuremath{V_{\mathrm{pr}}}}

\newcommand{\RR}{\ensuremath{\mathbb{R}}} 
\newtheorem{theorem}{Theorem}[section]

\newtheorem{lemma}[theorem]{Lemma}

\newtheorem{corollary}[theorem]{Corollary}

\theoremstyle{definition}

\theoremstyle{remark}

\newcommand{\eqnspatial}{1}
\newcommand{\secasymps}{3.2}
\newcommand{\secasymts}{3.3}
\newcommand{\thmpsest}{1}
\newcommand{\thmpslambda}{2}
\newcommand{\eqnpsestimates}{3}
\newcommand{\corps}{1}
\newcommand{\thmtsest}{3}
\newcommand{\thmtslambda}{4}
\newcommand{\eqntsestimatesb}{7}
\newcommand{\eqntsestimatesf}{8}
\newcommand{\corts}{2}
\newcommand{\sects}{2.2}
\newcommand{\eqnmod}{10}
\newcommand{\eqnmodplus}{14}
\newcommand{\eqnmodgsemtwo}{13}
\newcommand{\secspatialplusglm}{6.2}
\newcommand{\eqnspatialglm}{16}

\title{Spatial+: a novel approach to spatial confounding}
\author{Emiko Dupont$^{1,*}$, 
Simon N. Wood$^{2}$, and Nicole Augustin$^{2}$ \\
$^{1}$Department of Mathematical Sciences, University of Bath, Bath BA2 7AY, U.K \\
$^{2}$School of Mathematics, University of Edinburgh, Edinburgh EH9 3FD, U.K.\\
$^{*}$e.dupont@bath.ac.uk}
\date{}

\begin{document}
\maketitle

\begin{abstract}
In spatial regression models, collinearity between covariates and spatial effects can lead to significant bias in effect estimates. This problem, known as spatial confounding, is encountered modelling forestry data to assess the effect of temperature on tree health. Reliable inference is difficult as results depend on whether or not spatial effects are included in the model. The mechanism behind spatial confounding is poorly understood and methods for dealing with it are limited. We propose a novel approach, spatial+, in which collinearity is reduced by replacing the covariates in the spatial model by their residuals after spatial dependence has been regressed away. Using a thin plate spline model formulation, we recognise spatial confounding as a smoothing-induced bias identified by \citet{rice1986}, and through asymptotic analysis of the effect estimates, we show that spatial+ avoids the bias problems of the spatial model. 
This is also demonstrated in a simulation study. Spatial+ is straight-forward to implement using existing software and, as the response variable is the same as that of the spatial model, standard model selection criteria can be used for comparisons. A major advantage of the method is also that it extends to models with non-Gaussian response distributions. 
Finally, while our results are derived in a thin plate spline setting, the spatial+ methodology transfers easily to other spatial model formulations.
\end{abstract}

\section{Introduction}\label{sec:intro}
Regression models for spatially referenced data use spatial random effects to capture residual spatial correlation that cannot be explained by the covariates in the model. While such models are an effective tool for predictions of the response variable, as first noted by \citet{clayton1993}, they can be problematic when estimation of individual covariate effects are of interest. So-called spatial confounding arises because spatial effects may have elements of collinearity with spatially dependent covariates and therefore interfere with their effect estimates. \citet{reich2006} analysed the issue using an example modelling the effect of socio-economic status on stomach cancer incidence in the municipalities of Slovenia. When spatial effects are added to the model, the covariate effect disappears, suggesting the spatial effects have taken over a disproportionate part of the explanatory power. While in this example, the spatial effects take the form of an Intrinsic Conditional Auto-Regressive (ICAR) random effect, spatial confounding is widely acknowledged as an issue that affects spatial models in general \citep[see e.g.][]{hodges2010,paciorek}.

In this paper we model data from the Terrestrial Crown Condition Inventory (TCCI) forest health monitoring survey which has been carried out yearly since 1983 by the Forest Research Institute Baden-W\"urttemberg. Crown defoliation (an indicator of poor tree health) has generally been worsening over time, and there is growing interest in understanding the effects of climate change in order to decide on forest management strategies for mitigation. Here, using a linear regression model, we consider the effect of temperature on crown defoliation. However, our results are highly dependent on whether or not we include spatial random effects in the model. As illustrated in Figure \ref{fig:intro}, in the null model (with no spatial effects), the estimated covariate effect is positive but not significant, whereas in the corresponding spatial model, the covariate effect is significant and the effect size more than triples. This behaviour suggests there is spatial confounding and makes reliable inference difficult. 

\begin{figure}
\centerline{
\includegraphics[width=0.32\textwidth]{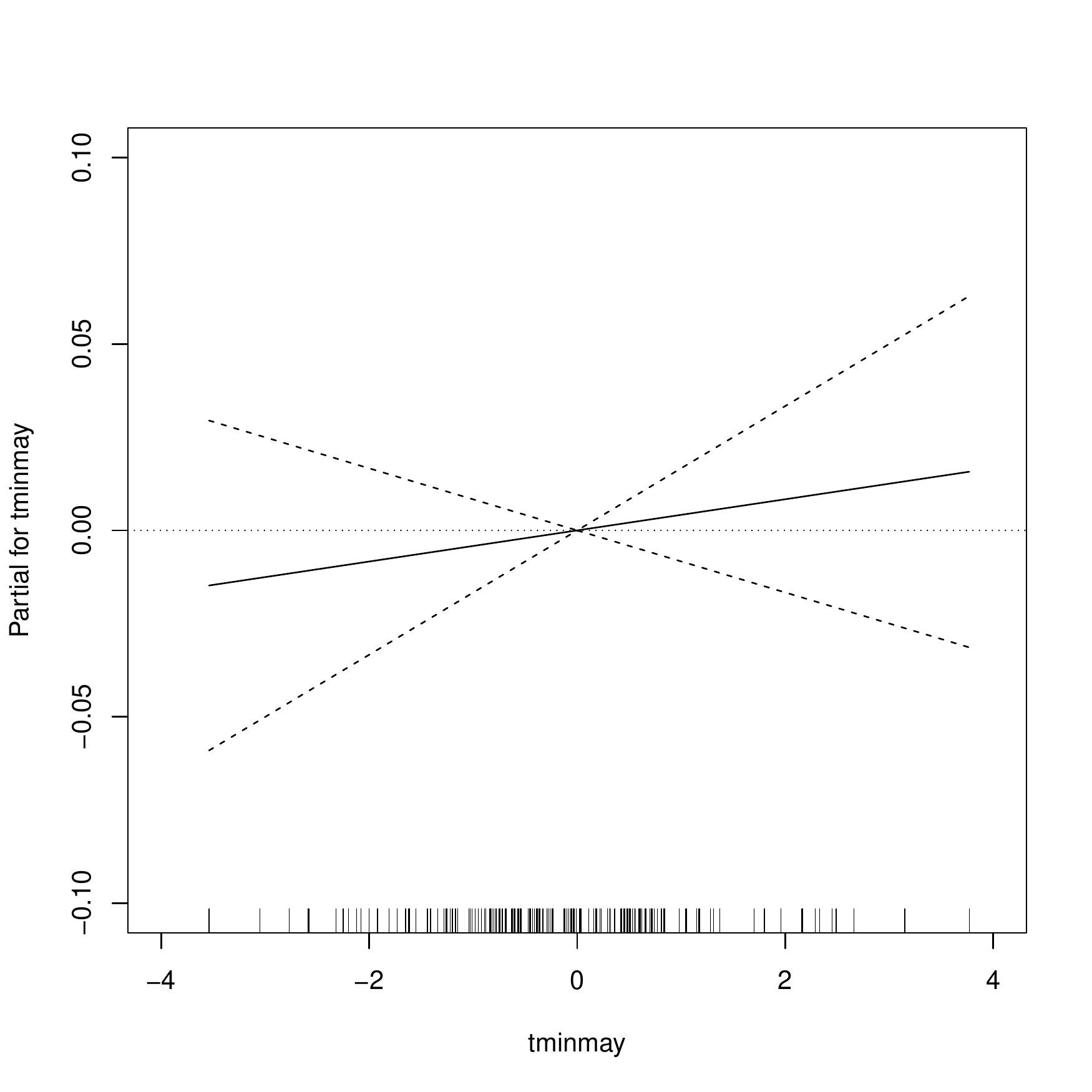}
\includegraphics[width=0.32\textwidth]{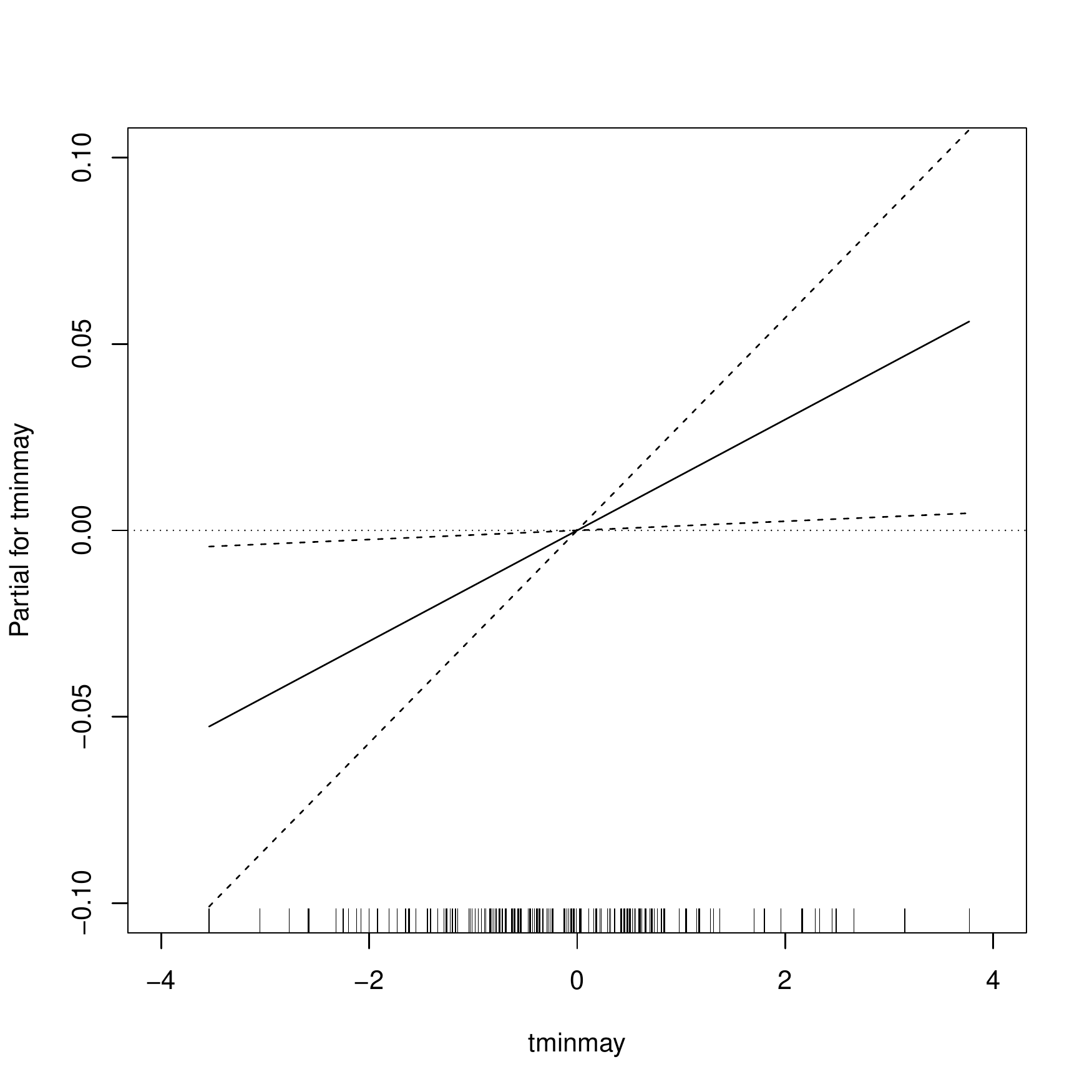}
\includegraphics[width=0.32\textwidth]{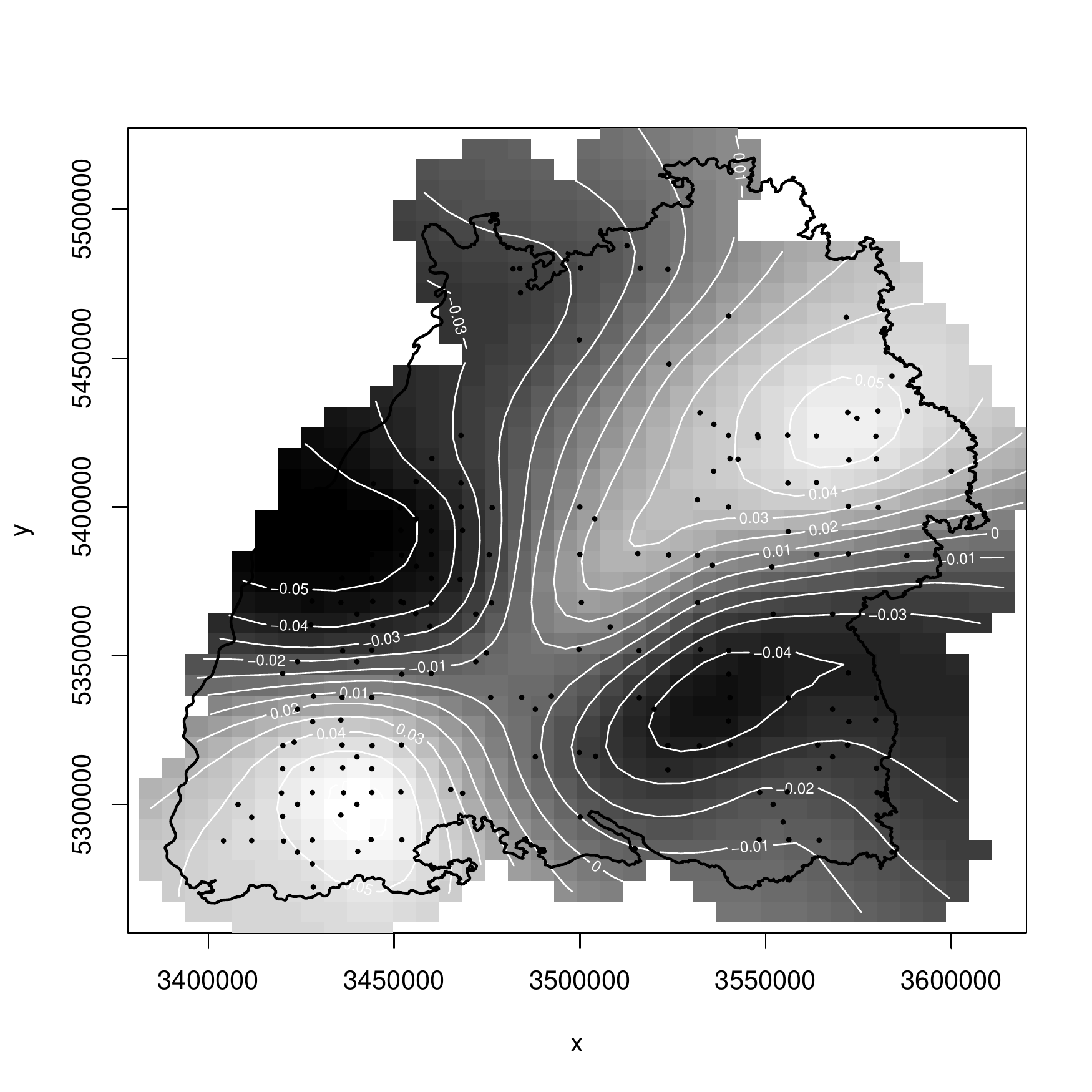}
}
\caption{Forestry example. Estimated effect of minimum temperature in May on crown defoliation in the null model (left) and the spatial model (middle), where for each model the plot shows the contribution of the centered covariate to the predicted response (with two times standard error bands). Estimated spatial effect in the spatial model (right) with the border of Baden-W\"urttemberg outlined and dots showing the data locations.}
\label{fig:intro}
\end{figure}

A commonly used method for dealing with spatial confounding is restricted spatial regression (RSR), introduced by \citet{reich2006} for the ICAR model, and further developed by \citet{hoeting} for continuous space models. In RSR the spatial random effects are restricted to the orthogonal complement of the covariates while keeping the overall column space of the model matrix in the regression unchanged. RSR directly eliminates collinearity and is designed to preserve the covariate effect estimates of the null model. However, it is a general misconception that the estimates in the null model (and thus, RSR) are unbiased. In fact, each covariate effect estimate includes, not only the effect of the covariate, but also any unmeasured spatial effects with the same spatial pattern \citep[see e.g.][]{hoeting}. Therefore, unless the unmeasured effects are independent of the covariates (a rather strong assumption), RSR introduces bias by construction.

Other methods for avoiding spatial confounding bias are limited, and with theoretical derivations often intractable, methodology tends to rely on simulations alone.
Recently, \citet{thaden2018} proposed the geoadditive structural equations model (gSEM) based on the structural equations framework commonly used for causal inference problems. In this approach, spatial dependence is regressed away from both the response and the covariates, and a regression involving the residuals only is used to identify the original covariate effects. 
In a simulation study, Thaden and Kneib show that the gSEM approach appears to give unbiased estimates while the original spatial model gives biased results in some scenarios. However, it is not immediately clear why the method works, and in examples such as ours, where the variables of interest are naturally spatially dependent, it seems undesirable to eliminate all spatial information from the modelling.

Here, we propose a novel approach, the spatial+ model, and using asymptotic analysis as well as a simulation study, we show that the estimates in this model avoid the bias problems of the spatial model. Spatial+ is a simple modification of the spatial model where the covariates are replaced by their residuals after spatial dependence has been regressed away. Thus, spatial+ retains the column space of the model matrix but reduces collinearity between covariate and spatial effects. A practical advantage over gSEM is that, as the response variable is unchanged from the spatial model, standard model selection criteria can be used for comparisons with the spatial and null models. Moreover, while the main properties of spatial+ are studied for models with a Gaussian response variable, we show that the method generalises naturally to any response distribution from the exponential family of distributions.

Key to our analysis is that we formulate the spatial model as a partial thin plate spline model where spatial correlation is modelled by imposing a smoothing penalty on the spatial effects in the fitting process. We can then recognise spatial confounding as a smoothing-induced bias identified by \citet{rice1986} for partial spline models (i.e.\ models where the domain of the spline, here the spatial domain, is one-dimensional). Spatial+ is a higher-dimensional version of a model introduced by \citet{chen_shiau1991} to overcome this type of bias in the one-dimensional case. Intuitively, as smoothing is only applied to the spatial part of the model, by making covariate effect estimates broadly independent of the spatial effects, they avoid this bias. For one-dimensional models, Rice, Chen and Shiau derived results for the asymptotic behaviour of estimates as the number of fitted data points $n\rightarrow\infty$. Due to results by \citet{utreras1988}, we are able to extend these derivations to models of arbitrary spatial dimension. We confirm that, as is the case in dimension one, the bias in the covariate effect estimates in the spatial model can become disproportionately large, while in the spatial+ model, the bias converges to $0$ strictly faster than the standard deviation.

Using the thin plate spline model formulation, we also see that, in the one-dimensional case, the gSEM estimates are in fact the same as the partial residual estimates introduced by Denby (1986) and independently by \citet{speckman1988}. Like the spatial+ estimates, the partial residual estimates (in dimension one) are shown by \citet{chen_shiau1991} to avoid the disproportionate smoothing-induced bias identified by \citet{rice1986}, and we can once again use \citet{utreras1988} to generalise this result to arbitrary spatial dimensions. For completeness, these derivations are included in Appendix D. 

Finally, we note that, although our results are derived in the thin plate spline context, the methodology of spatial+, namely, the modification of the model matrix, can be directly applied to other commonly used spatial models, including, for example, Gaussian Markov random field (GMRF) models and the (discrete space) ICAR model. In fact, it can be shown that modelling spatial random effects through the use of a smoothing penalty is equivalent to a Bayesian model formulation in which the spatial correlation structure is determined by a prior distribution. This equivalence is explained, for example, in \citet{kimeldorf1970}, Section 6.1 of \citet{silverman1985}, pages 239-240 of \citet{wood2017} and \citet{fahrmeir2004}. Thus, while different spatial models correspond to different smoothing penalties, the underlying idea of reducing collinearity in this way to keep covariate effect estimates unaffected by spatial smoothing would apply in general. 

This paper is structured as follows. In Section \ref{sec:method}, we introduce the spatial and spatial+ models that form the basis of our analysis. Section \ref{sec:asym} summarises the main theoretical results which show that the covariate effect estimates in the spatial model can become disproportionately biased due to spatial smoothing, while in the spatial+ model, bias is negligible. In Section \ref{sec:sim}, we illustrate these results in a simulation study which also compares spatial+ with RSR and the gSEM. In Section \ref{sec:app}, we demonstrate how spatial+ can be implemented by applying it to our forestry example. Finally, in Section \ref{sec:glm}, we generalise the spatial+ methodology to non-Gaussian response distributions and confirm that the method works in simulations for three different distributions.

\section{Method}\label{sec:method}
\subsection{Spatial model}\label{sec:ps}
Our starting point is a spatial model formulated as a partial thin plate spline model \citep[see e.g.][]{wahba1990} of the form
\begin{equation}\label{eqn:spatial}
y_i=\beta x_i+f(\tb_i)+\epsilon_i, \quad \epsilon_i\underset{\textrm{iid}}\sim N(0,\sigma^2)
\end{equation}
where $\yb=(y_1,\ldots,y_n)^T$ is the response, $\xb=(x_1,\ldots,x_n)^T$ an observed covariate with unknown effect $\beta$ and $f\in H^m(\Omega)$ an unknown bounded function defined on an open bounded domain $\Omega\subset\RR^d$ which includes the known values $\tb_1,\ldots,\tb_n$. In the spatial context, $\tb_1,\ldots,\tb_n$ are the spatial locations of the observations. The estimates $\betahps$ and $\hat{f}$ in this model (known as the partial thin plate spline estimates of order $m > d/2$) are obtained as the minimisers of
\[
\frac{1}{n}\sum_{i=1}^n (y_i-\beta x_i-f(\tb_i))^2+\lambda \sum_{i_1,\ldots,i_m}\int_{\RR^d} \big\vert \frac{\partial^m f(\tb)}{\partial t_{i_1}\cdots \partial t_{i_m}}\big\vert^2 d\tb
\]
where $\lambda>0$ is an unknown smoothing parameter. 
Minimisation here is over all $\beta\in \RR$ and functions $f\in H^m(\RR^d)$ with $\frac{\partial^m f}{\partial t_{i_1}\cdots \partial t_{i_m}}\in L^2(\RR^d)$ for all subsets $i_1,\ldots,i_m$ of $1,\ldots,n$. The first term encourages fitted values that are close to the data while the second term induces smoothing by penalising the wiggliness of the function $f$.

\citet{duchon1977} showed that the estimate of $f$ can be obtained by estimating its coefficients in a basis known as the natural thin plate spline basis. This basis spans a finite-dimensional subspace in the space of functions defined on $\RR^d$ and has dimension $N=M+n$ where $M={{m+d-1}\choose{d}}$. Using this basis, the partial thin plate spline estimates $\betahps$ and $\fbhps=(\hat{f}(\tb_1),\ldots,\hat{f}(\tb_n))^T$ are the minimisers of
\begin{equation}\label{eqn:ps_min}
\Vert\yb-\beta\xb-\fb\Vert^2+n\lambda\fb^T\boldsymbol{\boldsymbol{\Gamma}}\fb
\end{equation}
with $\boldsymbol{\Gamma}$ an $n\times n$ penalty matrix. Solving the resulting normal equations, we see that
\begin{equation}\label{eqn:ps_estimates}
\betahps=\left(\xb^T(\IB-\SBl)\xb\right)^{-1}\xb^T(\IB-\SBl)\yb,\quad \fbhps=\SBl(\yb-\betahps\xb)
\end{equation}
where $\SBl=(\IB+n\lambda\boldsymbol{\Gamma})^{-1}$ is known as the smoother matrix. $\SBl$ is the influence matrix for the model (\ref{eqn:spatial}) with no covariate term, and $\SBl\yb$ is the thin plate spline fitted to the data $\yb$.

\subsection{Spatial+ model}\label{sec:ts}
Starting with the model (\ref{eqn:spatial}), in line with \citet{rice1986}, we assume the covariate $\xb$ has the form 
\begin{equation}\label{eqn:co_model}
x_i=f^x(\tb_i)+\epsilon^x_i,\quad \epsilon^x_i\underset{\textrm{iid}}\sim N(0,\sigma_x^2)
\end{equation}
where $f^x\in H^m(\Omega)$ is bounded. This means that $\xb$ is correlated with the smooth term $f$ through the component $f^x$. Extending the two-stage smoothing spline model defined in \citet{chen_shiau1991} to models of dimension $d\ge 1$, we define the spatial+ model as follows. Let $\fbhps^x=\SBlx\xb$ and $\rb^x=(\IB-\SBlx)\xb$ be the fitted values and residuals in the thin plate spline regression (\ref{eqn:co_model}) with smoothing parameter $\lambda_x>0$. The spatial+ model is then the partial thin plate spline model
\begin{equation}\label{eqn:spatial+}
y_i=\beta r^x_i+f^+(\tb_i)+\epsilon_i, \quad \epsilon_i\underset{\textrm{iid}}\sim N(0,\sigma^2)
\end{equation}
where $\beta$ is the unknown effect of $\rb^x=(r^x_1,\ldots,r^x_n)^T$ and $f^+$ models the combined effect $f+\beta f^x$ in the original model (\ref{eqn:spatial}). The spatial+ estimate $\betahts$ of $\beta$ is its partial thin plate spline estimate in this model, i.e.\
\begin{align*}
&\betahts=\big((\rb^x)^T(\IB-\SBl)\rb^x\big)^{-1}(\rb^x)^T(\IB-\SBl)\yb&\\
&=\big(\xb^T(\IB-\SBlx)(\IB-\SBl)(\IB-\SBlx)\xb\big)^{-1}\!\xb^T(\IB-\SBlx)(\IB-\SBl)\yb,&\nonumber
\end{align*}
and the spatial+ estimate $\fbhts$ of $\fb=(f(\tb_1),\ldots,f(\tb_n))^T$ is given by
\[
\fbhts=\widehat{\fb^+}-\betahts \fbhps^x=\SBl(\yb-\betahts\xb)-(\IB-\SBl)\SBlx\betahts\xb
\]
where $\widehat{\fb^+}$ denotes the partial thin plate spline estimate of $\fb^+=(f^+(\tb_1),\ldots,f^+(\tb_n))^T$ in (\ref{eqn:spatial+}).

\subsection{Smoothness selection}
Smoothing penalties introduce bias in estimates but reduce variance. The smoothing parameters $\lambda$ and $\lambda_x$ are usually estimated based on a separate smoothness selection criterion, that balances this bias-variance trade-off.

For the analysis in Section \ref{sec:asym}, in line with \citet{rice1986, chen_shiau1991}, we choose the value of the smoothing parameter that minimises the average mean squared error (AMSE) of the estimated spatial effect. The AMSE for an estimated effect $\fbh=(\hat{f}_1,\ldots,\hat{f}_n)^T$ of the function $f$ evaluated at data points is defined as
\[
\amse(\fbh)=\frac{1}{n}\sum_{i=1}^n \E\left[(\hat{f}_i-f(\tb_i))^2\right]
=B^2(f,\lambda)+V(f,\lambda)
\]
where $B^2(f,\lambda)=\frac{1}{n}\sum_{i=1}^n \big(\E(\hat{f}_i)-f(\tb_i)\big)^2$ and $V(f,\lambda)=\frac{1}{n}\sum_{i=1}^n\var(\hat{f}_i)$ are the average squared bias and the average variance, respectively.

For the simulations in Section \ref{sec:sim} and the example in Section \ref{sec:app} we use the generalized cross validation (GCV) criterion, which is the default option in the \texttt{R}-package \texttt{mgcv} used for implementation. Asymptotically (as the sample size $n\rightarrow\infty$), GCV selects the optimal smoothing parameter for minimising prediction error. Thus, GCV is not dissimilar to the criterion used for the theoretical derivations. Indeed, \citet{chen_shiau1994} show that their asymptotic results in \citet{chen_shiau1991} for one-dimensional models also hold for GCV and Mallows' $C_L$. In practice, the restricted maximum likelihood (REML) criterion is often used instead of GCV as, for finite samples, GCV usually has more uncertain estimates than REML and tends to undersmooth (i.e.\ overfit) the data  (\citet{wood2017} p. 266-267). Repeating the simulations and the data example using REML gave similar results to GCV.

\section{Asymptotic results}\label{sec:asym}
In Sections \ref{sec:asym_ps} and \ref{sec:asym_ts} we derive asymptotic results for the models defined in Sections \ref{sec:ps} and \ref{sec:ts} with proofs provided in Appendix C. These results are based on a number of technical lemmas and assumptions, details of which are provided in Appendices A and B. For dimension $d=1$, \citet{rice1986}, \citet{chen_shiau1991} use the Demmler-Reinsch basis for natural splines to diagonalise the smoother matrix $\SBl$. This enables them to explicitly study the asymptotic behaviour of model estimates. Our generalisation to dimensions $d\ge 1$ is fascilitated by \citet{utreras1988} who shows how the asymptotic properties of the Demmler-Reinsch basis generalise to higher dimensions (see Lemma 1 in Appendix A). 
Using this, we are able to prove results in a similar way to Rice, Chen and Shiau. Where possible, we have simplified the derivations and notation, particularly, we have adapted the structure of Rice's proofs to be more in line with the approach used by \citet{chen_shiau1991}. For the rest of this paper, we assume that $m> d/2$ and that the domain $\Omega$ and data locations $\tb_1,\ldots,\tb_n$ satisfy the conditions of Lemma 1 in Appendix A. 
We use the notation $a(n)\approx b(n)$ to mean that $a(n)/b(n)$ is bounded away from zero and infinity as $n\rightarrow \infty$.

\subsection{Asymptotic results for the spatial model}\label{sec:asym_ps}
In the model (\ref{eqn:spatial}), spatial correlation is modelled through smoothing of the term $f$. Without the smoothing penalty, the model is an ordinary linear model in which all effect estimates are unbiased. Therefore, bias in the covariate effect estimate arises as a direct result of smoothing. \citet{rice1986} showed for dimension $d=1$ that, while this bias is asymptotically $0$ as $n\rightarrow\infty$, the rate of convergence may be slow. More specifically, we cannot ensure that the bias converges faster than the standard deviation if the smoothing parameter $\lambda$ converges at the optimal rate (minimising the AMSE of the estimated spatial effect). Therefore, the bias can in practice become disproportionately large. Here, we generalise Rice's results and see that the problem of potentially excessive bias in $\betahps$ persists in models where the spatial domain has dimension $d\ge 1$. As an aside, we note that, as in the $d=1$ case, the rate of convergence of the variance of $\betahps$, is the same as that in a model with no smoothing penalty.

\begin{theorem}\label{thm:ps_est}
Suppose $\lambda\approx n^{-\delta}$ for some $0<\delta<1$, $f, f^x\in H^m(\Omega)$ are bounded and $m\ge d$.  Then for the partial thin plate spline estimate of $\beta$ we have that
\begin{description}
\item[(a)] $\E(\betahps)-\beta=\lo(n^{-1/2})+\LO(\lambda^{1/2})$,
\item[(b)] $n\var(\betahps)\rightarrow \sigma^2/\sigma_x^2$ as $n\rightarrow\infty$.
\end{description}
In particular, $\var(\betahps)=\LO(n^{-1})$ and we need $\lambda=\lo(n^{-1})$ to ensure that the bias converges faster than the standard deviation of $\betahps$.
\end{theorem}

\begin{theorem}\label{thm:ps_lambda}
Suppose $\lambda\approx n^{-\delta}$ for some $0<\delta<1$, $f, f^x\in H^m(\Omega)$ are bounded and $m\ge d$. Then the average squared bias $\Bps(f,\lambda)$ and average variance $\Vps(f,\lambda)$ of the partial thin plate spline estimate of $f$ satisfy
\begin{description}
\item[(a)] $\Bps(f,\lambda)=n^{-1}\sum_i (\E(\fbhpsi)-f(\tb_i))^2=\LO(\lambda)$,
\item[(b)] $\Vps(f,\lambda)=n^{-1}\sum_i\var (\fbhpsi)=\LO(n^{-1}\lambda^{-d/2m})$.
\end{description}
In particular, the optimal rate for $\lambda$ in terms of minimising $\amse(\fbhps)$ is $\lambda=\LO(n^{-2m/(2m+d)})$, and when $\lambda$ converges at this optimal rate, $\amse(\fbhps)=\LO(n^{-2m/(2m+d)})$.
\end{theorem}

We have therefore proved the following result which shows that we cannot avoid the potential for excessive bias in $\betahps$, unless $\lambda$ converges at a rate slower than the optimal rate of convergence, i.e.\ unless the smooth term is undersmoothed.
\begin{corollary}\label{cor:ps}
Suppose $\lambda\approx n^{-\delta}$ for some $0<\delta<1$, $f, f^x\in H^m(\Omega)$ are bounded and $m\ge d$.
The optimal rate of convergence for $\lambda$ in terms of minimising $\amse(\fbhps)$ is slower than the required rate of $\lo(n^{-1})$ for ensuring that the bias of $\betahps$ converges faster than the standard deviation of the estimate.
\end{corollary}

\subsection{Asymptotic results for the spatial+ model}\label{sec:asym_ts}
In dimension $d=1$, \citet{chen_shiau1991} show that for the model (\ref{eqn:spatial+}), the problems identified by Rice disappear. That is, when the parameters $\lambda$ and $\lambda_x$ converge at the optimal rate (for minimising the AMSE of the estimated spatial effect), the bias of the covariate effect estimate $\betahts$ converges to $0$ faster than the standard deviation and, therefore, does not become disproportionately large. We now generalise these results to dimensions $d\ge 1$.
\begin{theorem}\label{thm:ts_est}
Suppose $\lambda\approx n^{-\delta}, \lambda_x\approx n^{-\delta_x}$ for some $0<\delta,\delta_x<1$, $f, f^x\in H^m(\Omega)$ are bounded and $m\ge d$. Then for the spatial+ estimate of $\beta$ we have that
\begin{description}
\item[(a)] $\E(\betahts)-\beta=\lo(n^{-1/2})+\LO((\lambda\lambda_x)^{1/2})$,
\item[(b)] $n\var(\betahts)\rightarrow\sigma^2/\sigma_x^2$ as $n\rightarrow\infty$.
\end{description}
In particular, $\var(\betahts)=\LO(n^{-1})$ and we need $\lambda\lambda_x=\lo(n^{-1})$ to ensure that the bias converges faster than the standard deviation of $\betahts$.
\end{theorem}

\begin{theorem}\label{thm:ts_lambda}
Suppose $\lambda\approx n^{-\delta}, \lambda_x\approx n^{-\delta_x}$ for some $0<\delta,\delta_x<1$, $f, f^x\in H^m(\Omega)$ are bounded and $m\ge d$.
Then the average squared bias $\Bts(f,\lambda,\lambda_x)$ and average variance $\Vts(f,\lambda,\lambda_x)$ of the spatial+ estimate of $f$ satisfy
\begin{description}
\item[(a)] $\Bts(f,\lambda,\lambda_x)=n^{-1}\sum_i (\E(\fbhtsi)-f(\tb_i))^2=\LO(\lambda)+\LO(n^{-1}\lambda_x^{-d/2m}\log^2 n)$,
\item[(b)] $\Vts(f,\lambda,\lambda_x)=n^{-1}\sum_i\var (\fbhtsi)=\LO(n^{-1}\lambda^{-d/2m})$.
\end{description}
In particular, the optimal rates for $\lambda$ and $\lambda_x$ in terms of minimising $\amse(\fbhts)$ are given by $\lambda=\LO(n^{-2m/(2m+d)})$ and $\lambda_x=\LO(n^{-2m/(2m+d)}(\log n)^{4m/d})$, assuming the convergence rates for $\Bts(f,\lambda,\lambda)$ and $\Vts(f,\lambda,\lambda_x)$ are equal. When $\lambda$ and $\lambda_x$ converge at these rates, $\amse(\fbhts)=\LO(n^{-2m/(2m+d)})$.
\end{theorem}
From this we obtain the following result which shows that, unlike $\betah$, the estimate $\betahts$ does not need undersmoothing to avoid excessive bias.
\begin{corollary}\label{cor:ts}
Suppose $\lambda\approx n^{-\delta}, \lambda_x\approx n^{-\delta_x}$ for some $0<\delta,\delta_x<1$, $f, f^x\in H^m(\Omega)$ are bounded and $m\ge d$. If $\lambda$ and $\lambda_x$ converge at the optimal rates in terms of minimising $\amse(\fbhts)$, then
$\lambda\lambda_x=\lo(n^{-1})$. In particular, the optimal rates for $\lambda$ and $\lambda_x$ ensure that the bias of the spatial+ estimate $\betahts$ converges faster than the standard deviation of the estimate. 
\end{corollary}

\section{Simulation}\label{sec:sim}
Partial thin plate spline models can be implemented in the $\texttt{R}$-package $\texttt{mgcv}$ using the computationally efficient reduced rank approximation known as thin plate regression splines. We use this implementation (with GCV as the smoothness selection criterion) to compare the results of models fitted to simulated data for which we know the true underlying covariate and spatial dependence.
\subsection{Data}\label{sec:sim_data}
We generate 100 independent replicates of covariate data $\xb=(x_1,\ldots,x_n)^T$ and response data $\yb=(y_1,\ldots,y_n)^T$, observed at $n=1000$ randomly selected locations in the spatial domain $[0,10]\times[0,10]$ in $\RR^2$ (using a $50\times50$ grid), as follows. Let $\zb=(z_1,\ldots,z_n)^T$ and $\zb'=(z'_1,\ldots,z'_n)^T$ denote observations at the selected locations of independently generated Gaussian spatial fields with an exponential and a spherical covariance structure, respectively. That is, each spatial field is sampled from a multivariate normal distribution centered at $\boldsymbol{0}$ with covariance structure defined by $C(h)=\exp(-(h/R)^p)$ with $R=5$ and $p=1$ for the exponential field and $C(h)=-1-1.5h/R+0.5(h/R)^3$ for $h\le R$, $C(h)=0$ for $h>R$ with $R=1$ for the spherical field (where $h$ denotes Euclidean distance). To ensure that the fields lie in the span of the spatial basis vectors used for the models in Section \ref{sec:sim_models}, each is replaced by the fitted values of a spatial thin plate regression spline fitted to them. 
We then let
\begin{eqnarray*}
\xb&=&0.5\zb +\eepsilon^x\textrm{ where }\eepsilon^x\sim N(\boldsymbol{0},\sigma_x^2\IB),\\
\yb&=&\beta\xb+\fb +\eepsilon^y\textrm{ where }\eepsilon^y\sim N(\boldsymbol{0},\sigma_y^2\IB),
\end{eqnarray*}
with true covariate effect $\beta=3$, true residual spatial effect $\fb=-\zb-\zb'$
and $\sigma_y=1$, $\sigma_x=0.1$. Thus, $\fb$ is directly correlated with the spatial pattern $0.5\zb$ of the covariate. This approach is similar to \citet{thaden2018}, except we have added the component $-\zb'$ so that $\fb$ could represent, for example, the combined effect of an unobserved covariate (with a similar spatial pattern to that of $\xb$) as well as an independent short-range spatial process.
Also, rather than treating the spatial fields as fixed, we generate new fields for each replicate in the simulation. Finally, we have chosen $\sigma_x$ relatively small (such that the model matrix for the spatial model has nearly collinear columns) and $\sigma_y$ relatively large (to encourage smoothing). This is the situation in which we would expect spatial confounding issues to arise which is also confirmed by the simulations in \citet{thaden2018}.

\subsection{Models}\label{sec:sim_models}
To each replicate of simulated response data $\yb$ and covariate data $\xb$, we fit the following models (with basis size $k=300$ for the thin plate regression splines). Models 2 - 5 are fitted twice: once with and once without smoothing penalties applied. In \texttt{mgcv}, smoothing penalties are applied by default but can be removed using the option \texttt{fx=TRUE}.
\begin{enumerate}
\item[1.] Null model: The model with no spatial effects given by
\begin{equation}\label{eqn:mod_null}
y_i=\beta x_i +\epsilon_i,\quad \eepsilon\sim N(\boldsymbol{0},\sigma^2\IB)
\end{equation}
where $\beta$ and $\sigma^2$ are estimated parameters.
\item[2.] Spatial model: The model given by
\begin{equation}\label{eqn:mod}
y_i=\beta x_i +f(\tb_i)+\epsilon_i,\quad \eepsilon\sim N(\boldsymbol{0},\sigma^2\IB)
\end{equation}
where $\beta$ and $\sigma^2$ are estimated parameters and $f$ a thin plate regression spline with $\tb_1,\ldots,\tb_n$ the observed data locations.
\item[3.] RSR model:
Let $\BBs$ be the matrix whose columns are the spatial basis vectors in the model matrix from (\ref{eqn:mod}) (i.e.\ the thin plate regression spline basis functions evaluated at the data locations) and let 
$\BBst=(\IB-\xb(\xb^T\xb)^{-1}\xb^T)\BBs$ 
be the projection of this onto the orthogonal compliment of $\xb$. The RSR model is given by
\begin{equation}\label{eqn:mod_rsr}
y_i=\beta x_i +\tilde{f}_i+\epsilon_i,\quad \eepsilon\sim N(\boldsymbol{0},\sigma^2\IB)
\end{equation}
where $\beta$ and $\sigma^2$ are estimated parameters and $\tilde{\fb}=(\tilde{f}_1,\ldots,\tilde{f}_n)^T$ is modelled the same way as the spatial effect in (\ref{eqn:mod}) but with $\BBs$ replaced by $\BBst$ in the model matrix.
\item[4.] gSEM: Let $\rb^x=(r_1^x,\ldots,r_n^x)^T$ and $\rb^y=(r_1^y,\ldots,r_n^y)^T$ denote the spatial residuals of $\xb$ and $\yb$, that is, $\rb^x=\xb-\hat{\fb}^x$ where $\hat{\fb}^x$ are the fitted values in the regression
\begin{equation}\label{eqn:mod_gsem1}
x_i=f^x(\tb_i) +\epsilon^x_i,\quad \eepsilon^x\sim N(\boldsymbol{0},\sigma_x^2\IB)
\end{equation}
where $\sigma_x^2$ is estimated and $f^x$ a thin plate regression spline, and $\rb^y$ is the same but replacing $\xb$ by $\yb$. The gSEM model is then the linear model given by
\begin{equation}\label{eqn:mod_gsem2}
r^y_i=\beta r^x_i +\epsilon_i,\quad \eepsilon\sim N(\boldsymbol{0},\sigma^2\IB),
\end{equation}
where $\beta$ and $\sigma^2$ are estimated.
\item[5.] Spatial+: Let $\rb^x$ denote the spatial residuals of $\xb$ as above. The spatial+ model is then
\begin{equation}\label{eqn:mod_plus}
y_i=\beta r^x_i +f^+(\tb_i)+\epsilon_i,\quad \eepsilon\sim N(\boldsymbol{0},\sigma^2\IB)
\end{equation}
where $\beta$ and $\sigma^2$ are estimated parameters and $f^+$ a thin plate regression spline.
\end{enumerate}

\subsection{Results}\label{sec:sim_results}
The results of the simulation are summarised in Figure \ref{fig:sim_results}. For each data replicate, the output is the estimated covariate effect and the mean squared error (MSE) of fitted values for each model fit. For ease of notation, in this section, we use $\betah$ to mean the estimated covariate effect in any of the fitted models (rather than the partial thin plate spline estimate alone). The MSE of fitted values is calculated as $\Vert\hat{\yb}-(\beta\xb+\fb)\Vert^2$ where for models 1, 2, 3 and 5, $\hat{\yb}$ is the fitted values in the regressions (\ref{eqn:mod_null}), (\ref{eqn:mod}), (\ref{eqn:mod_rsr}) and (\ref{eqn:mod_plus}), respectively, and for model 4, $\hat{\yb}=\fbh^y+\hat{\rb}^y$ where $\fbh^y$ and $\hat{\rb}^y$ are the fitted values in the regressions (\ref{eqn:mod_gsem1}) and (\ref{eqn:mod_gsem2}). Here $\beta=3$ and $\fb=-\zb-\zb'$ are the true values of the estimated effects with $\beta\xb+\fb$ the true mean of $\yb$.
\begin{figure}
\begin{tabular}{cc}
\includegraphics[width=0.48\textwidth]{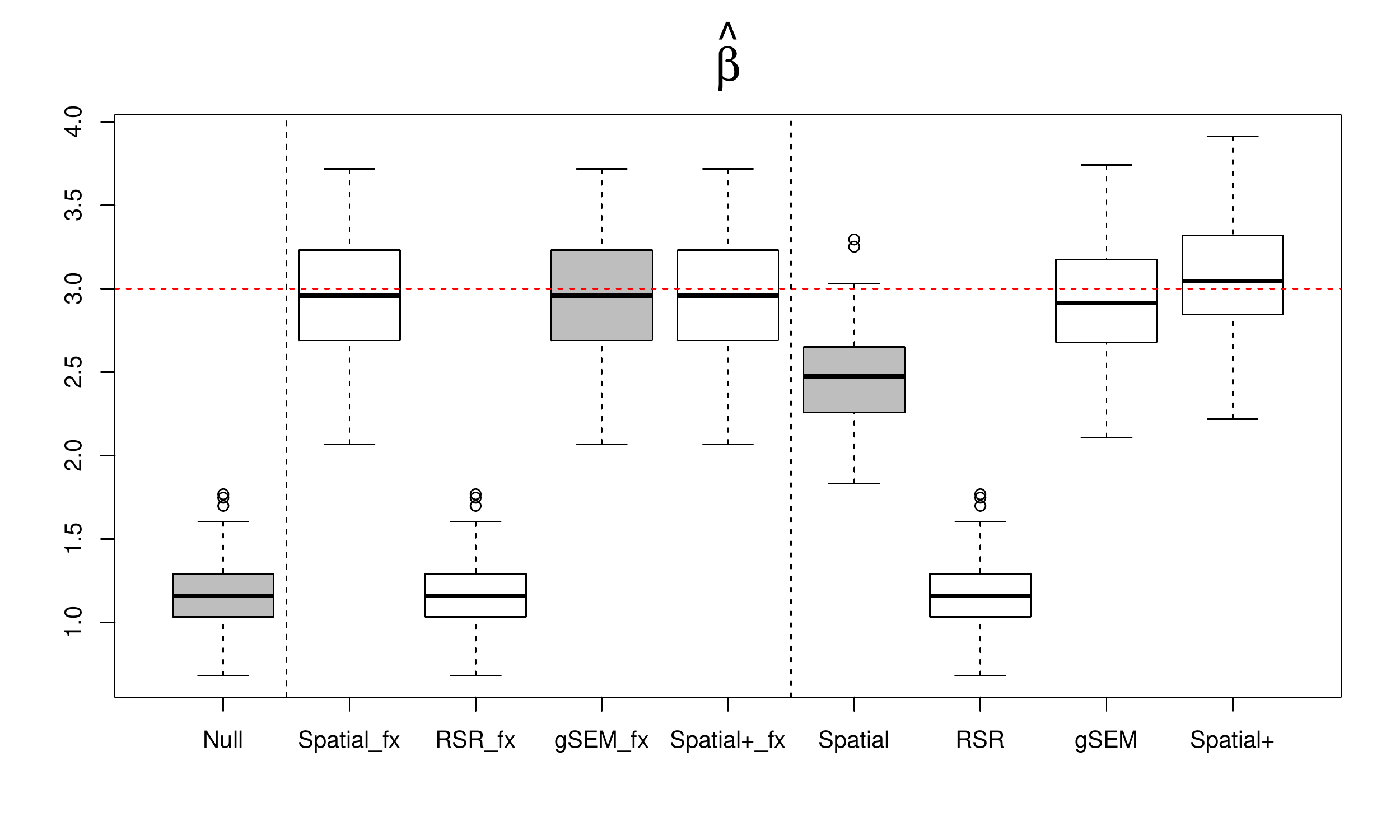}&
\includegraphics[width=0.48\textwidth]{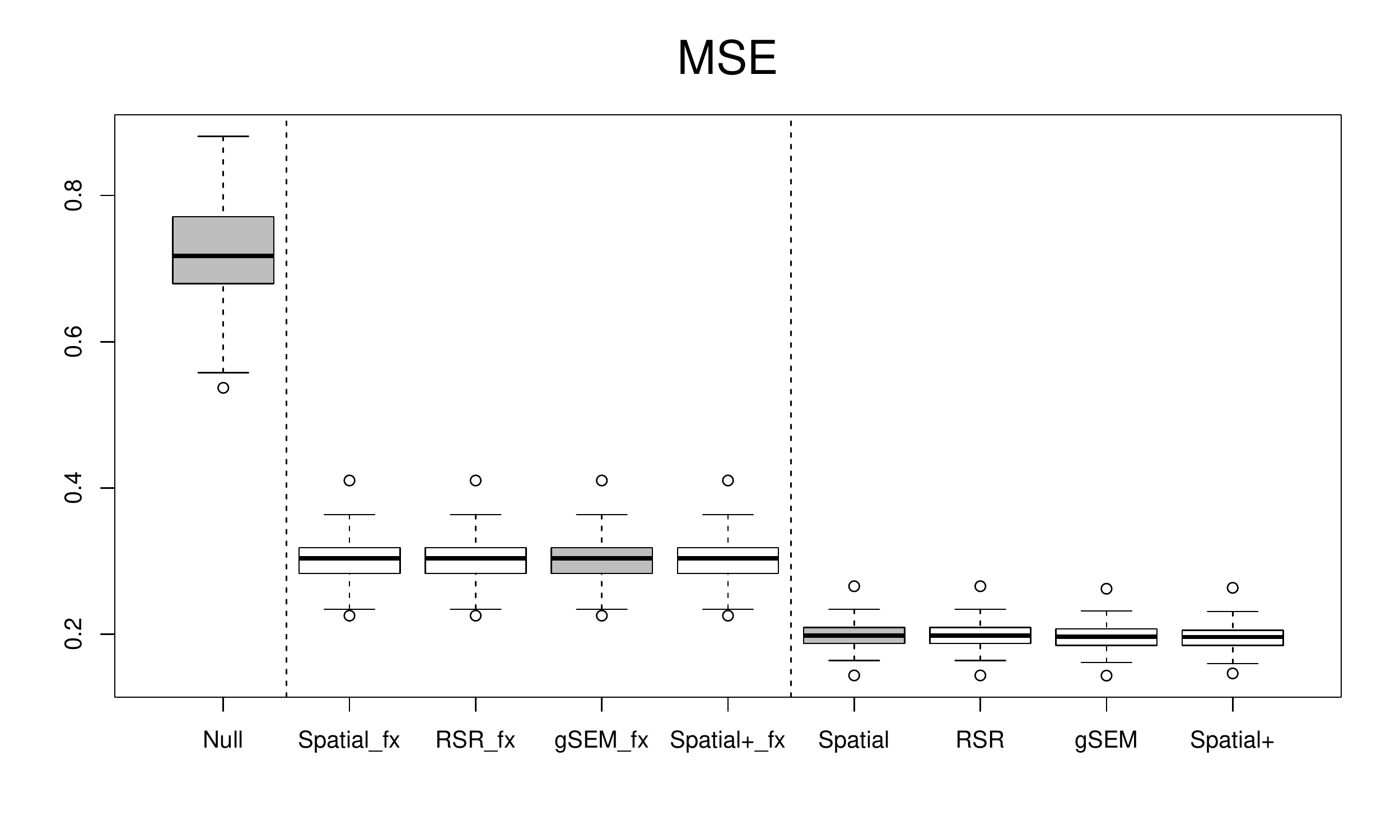}
\end{tabular}
\bigskip
\caption{Estimated covariate effect $\hat{\beta}$ (left) and MSE of fitted values (right) for each model fitted to 100 data replicates, where the true covariate effect is $\beta=3$. Subscript \texttt{fx} refers to a model in which no smoothing penalties were applied whereas no subscript models were smoothed. Results in grey are the three models that correspond to those used in Thaden and Kneib's simulation study.}\label{fig:sim_results}
\end{figure}

In the null model and the RSR model, the estimated covariate effect is the same and has a noticeably larger bias than the estimates in the other models. This is expected as for these models, $\betah=(\xb^T\xb)^{-1}\xb^T\yb$ is the ordinary least squares estimate, which, in addition to the true effect $\beta$, includes a contribution from the part of $\fb$ that is correlated with $\xb$. The fitted values in RSR, however, differ from those of the null model as the larger model matrix explains a part of $\yb$ that is treated as random noise in the null model. In fact, the column space of the model matrix is the same as that of the spatial model, and it is therefore not surprising that the fitted values in these two models are similar.

If no smoothing penalty is applied, models 2, 4 and 5 are essentially the same: they have the same fitted values and the same unbiased estimate for the covariate effect. This illustrates that spatial confounding bias is due to the combined effect of collinearity and smoothing, rather than collinearity alone. The spatial model is, in this case, an ordinary linear model where the columns in the model matrix are the covariate $\xb$ and the spatial basis vectors $\BBs$. This is the model from which the data is generated and, therefore, it is not surprising that the spatial model is able to recapture the true effects. The spatial+ model is a reparametrisation of the spatial model which preserves the overall column space, and simple linear model theory shows that the covariate effect estimate is preserved. Similarly, it is straightforward to show that the gSEM covariate effect estimate agrees with the spatial model in this case. (Derivations are included in Appendix E).

In the unsmoothed versions of models 2 - 5, the fitted values are all the same. When smoothing is applied, the MSE of fitted values reduces, indeed, this is the intended purpose of the smoothing penalty. Looking at the covariate effect estimate, in the RSR model, the (biased) ordinary least squares estimate is unaffected by smoothing. For the remaining three models, while the unsmoothed versions of the models give unbiased estimates of $\beta$, we see that smoothing introduces varying degrees of bias. In the spatial model, the bias is quite large illustrating our results in Section \ref{sec:asym_ps}. In contrast, while the covariate effect estimate is no longer the same in the gSEM and the spatial+ model, for both models, the bias is still negligible. This behaviour is therefore also consistent with what we would expect from our theoretical results.

Note that our analysis gives some intuition for why spatial+ works. If no smoothing penalty is applied, we saw that spatial+ has the same unbiased estimate for the covariate effect as the spatial model. In fact, any decomposition $\xb=\vb+\rb$ with $\vb$ in the column space of the spatial basis vectors $\BBs$ gives a reparametrisation (replacing $\xb$ by $\rb$) in which $\rb$ captures the original covariate effect. However, by choosing $\rb$ to be broadly orthogonal to the column space of $\BBs$ (as it is in the spatial+ model), the estimates of the covariate and spatial effects are broadly decorrelated. Thus, the covariate effect estimate is largely unaffected when smoothing is applied to the spatial term and thereby remains broadly unbiased.

\section{Application}\label{sec:app}
We illustrate how the spatial+ model can be used in practice by applying it to our forestry example. Details of the data can be found in \citet{augustin2009, eichhorn2016}. 
We consider here the data for spruce for a single observation year, namely, 2013 which has measurements from $n=186$ locations. We are interested in assessing the effect of the climate variable $\texttt{tminmay}$ (minimum temperature in May) on the response variable \texttt{ratio} (crown defoliation expressed as a proportion). We expect a high minimum temperature in May to be indicative of a warmer and drier year in general which, in turn, is likely to lead to higher levels of tree defoliation (measured later in summer). We also expect older trees to have significantly more defoliation than younger trees and have therefore included the variable \texttt{age} (age of trees) as an additional covariate in the models. Scatterplots of the data (not shown here) indicate the relationships between the covariates and the response variable are broadly as expected.

\subsection{Models}
A natural starting point is the null model
\begin{equation}\label{eqn:defol_null}
\texttt{ratio}_i=\alpha+\beta_1\texttt{age}_i+\beta_2\texttt{tminmay}_i+\epsilon_i, 
\end{equation}
where $\epsilon_i\sim N(0,\sigma^2)$ is iid noise and $\alpha,\beta_1,\beta_2$ and $\sigma$ are estimated parameters. However, numerous spatially dependent predictors have not been included in the model, for example, soil characteristics such as soil depth and base saturation; other climatic variables such as those related to radiation and precipitation; water budget of the trees etc. Therefore, we would expect residual spatial correlation in the response variable, and a more appropriate model may therefore be a spatial model, which we define as 
\[
\texttt{ratio}_i=\alpha+\beta_1\texttt{age}_i+\beta_2\texttt{tminmay}_i+f(\tb_i)+\epsilon_i,
\]
where $\epsilon_i\sim N(0,\sigma^2)$ is iid noise, $\alpha,\beta_1,\beta_2$ and $\sigma$ are estimated parameters and $f$ a thin plate regression spline (with basis size $k=100$) with $\tb_1,\ldots,\tb_n$ the observed data locations.

The covariate effects of interest are $\beta_1$ and $\beta_2$ but, as the results of Sections \ref{sec:asym_ps} and \ref{sec:sim} show, the estimates of these effects may be highly biased in both the null model and the spatial model. This disproportionate bias is avoided in the spatial+ model. Let $\rb^1=(r_1^{1},\ldots,r_n^{1})^T$ and $\rb^2=(r_1^{2},\ldots,r_n^{2})^T$ be the residuals when a thin plate regression spline (with basis size $k=100$) is fitted to $\texttt{age}$ and $\texttt{tminmay}$, respectively. The spatial+ model is then
\[
\texttt{ratio}_i=\alpha+\beta_1r^{1}_i+\beta_2r^{2}_i+f^+(\tb_i)+\epsilon_i, 
\]
where $\epsilon_i\sim N(0,\sigma^2)$ is iid noise, $\alpha,\beta_1,\beta_2$ and $\sigma$ are estimated parameters and $f^+$ a thin plate regression spline (with basis size $k=100$) with $\tb_1,\ldots,\tb_n$ the observed data locations. 

Finally, for comparison, we fit the gSEM as an alternative method for avoiding spatial confounding bias. Let $\rb^{y}=(r_1^{y},\ldots,r_n^{y})^T$ be the residuals when a thin plate regression spline (with basis size $k=100$) is fitted to the response variable $\texttt{ratio}$. The gSEM is then
\[
r^{y}_i=\beta_1r^{1}_i+\beta_2r^{2}_i+\epsilon_i, 
\]
where $\epsilon_i\sim N(0,\sigma^2)$ is iid noise and $\beta_1,\beta_2$ and $\sigma$ are estimated parameters.

\subsection{Results}
The results of fitting the above four models to the data are summarised in Table \ref{tab:defol_res}.
\begin{table}
\scriptsize
\setlength{\tabcolsep}{5pt}
\renewcommand{\arraystretch}{1.2}
\begin{tabular}{|l|rrr|rrr|rrr|rrr|}
  \hline& \multicolumn{3}{|c|}{\texttt{age}}& \multicolumn{3}{|c|}{\texttt{tminmay}}& \multicolumn{3}{|c|}{\texttt{s(x,y)}}&Dev&&\\ & $\betah$ & p-value &   & $\betah$ & p-value &   & edf & p-value &   & expl & $\hat{\sigma}$ & AIC \\ 
  \hline
Null & 0.00247 & $ <10^{-16}$ & *** & 0.0042 & 0.5049 &   &  &  &  & 0.490 & 0.00940 & -335 \\ 
  Spatial & 0.00237 & $ <10^{-16}$ &  *** & 0.0149 & 0.0307 &  * & 14.2 & 0.0243 &  * & 0.605 & 0.00789 & -355 \\ 
  Spatial+ & 0.00237 & $ <10^{-16}$ & *** & 0.0316 & 0.0073 & ** & 12.0 & 3.32e-05 &  *** & 0.598 & 0.00793 & -356 \\ 
  gSEM & 0.00232 & $ <10^{-16}$ &  *** & 0.0317 & 0.0058 &  ** &  &  &  & & & \\
   \hline
\end{tabular}
\caption{Forestry example: results of fitting models to the data. For each covariate: the estimate of the covariate effect $\beta$ and its p-value. $\texttt{s(x,y)}$ refers to the thin plate regression splines fitted to $f$ in the spatial model and $f^+$ in the spatial+ model. For each of these: the effective degrees of freedom (edf) and the p-value. For each significant p-value we write ‘***’ if it is $<0.001$, ‘**’ if $<0.01$ and ‘*’ if $<0.05$. Note that in the gSEM, deviance explained, estimated standard deviation and AIC do not compare directly with the other models as the response variable is different.}\label{tab:defol_res}
\end{table}

The spatial term in the spatial model is significant, which confirms there is residual spatial correlation in the data as expected. Furthermore, as the spatial term allows for more of the residual variation to be explained, the deviance explained is higher and the estimated standard deviation is lower than in the null model. As the AIC is also lower, we conclude that the spatial model is an overall better fitting model than the null model for this data. However, while the spatial model may be appropriate for overall predictions of the response variable, the estimate of any individual covariate effect may be biased. Using the spatial+ model, we expect to obtain similar fitted values as the spatial model but with covariate effect estimates that have only negligible bias. Indeed, in terms of overall fit, we see that the deviance explained, estimated standard deviation and AIC in the spatial+ model are similar to those of the spatial model. For completeness, we have also included the gSEM. Note, however, that in the gSEM, since the response variable in the regression differs from that of the other three models, the deviance explained, estimated standard deviation and AIC cannot be directly compared to the other models. 

The covariate \texttt{age} is highly significant and has a positive effect as expected. This covariate does not appear to be affected by spatial confounding as the estimated effect and its p-value are largely robust to the choice of model. This happens, for example, if a covariate is independent of the true underlying residual spatial effect. Also, in the case of \texttt{age}, not only is this a covariate that is not very well explained by spatial location (a spatial smooth fitted to this variable has deviance explained of only $13\%$), but its estimated spatial pattern looks dominated by linear spatial basis functions which are unpenalised in the spatial model. Therefore, penalisation of the spatial term $f$ in the spatial model is less likely to interfere with the covariate effect estimate (see \citet{rice1986} Proposition D).

In contrast, the estimated effect of the covariate $\texttt{tminmay}$ is not significant in the null model but is significant in the spatial model and is even more significant in the spatial+ model. Furthermore, while in all models the effect estimate is positive as expected (i.e.\ higher temperature in May leads to more defoliation later in summer), the size of the estimate more than triples when a spatial effect is added to the null model and the estimate in the spatial+ model is more than double of that in the spatial model. This shows that, if we were to use the spatial model for our inference, the effect of temperature on crown defoliation would likely be underestimated in both size and significance due to spatial confounding. Note that, as expected, the gSEM gives similar results to spatial+.

\section{Non-Gaussian response data}\label{sec:glm}
A distribution is in the exponential family of distributions if its probability density function $p$ can be written in the form
\[
p(y)=\exp \big[ \{y\theta-b(\theta)\}/a(\phi)+c(y,\phi)\big]
\]
where $\theta$ and $\phi$ are parameters of the distribution and $a,b$ and $c$ are functions. This family includes a large number of commonly used distributions in applied statistics, e.g.\ Gaussian, Poisson, exponential and binomial.
\subsection{Spatial model}\label{sec:glm_spatial}
Suppose we have response data $\yb=(y_1,\ldots,y_n)^T$ where each $y_i$ is assumed to be an observation of a random variable $Y_i$ whose distribution is from the exponential family with $\E(Y_i)=\mu_i$, and suppose $\xb=(x_1,\ldots,x_n)^T$ and $\tb_1,\ldots,\tb_n$ are covariate observations and spatial locations as before. A generalised version of (\ref{eqn:spatial}) can then be formulated as
\begin{equation}\label{eqn:spatial_glm}
g(\mu_i)=\beta x_i+f(\tb_i)
\end{equation}
where $\beta$ is an unknown parameter, $f$ a thin plate spline and $g:\RR\rightarrow\RR$ a link function (i.e.\ a monotonic smooth function which ensures $g(\mu_i)$ is in the domain of the response variable). The partial thin plate spline estimates of $\beta$ and $\fb=(f(\tb_1),\ldots,f(\tb_n))^T$ are found using a penalised iterative re-weighted least squares (PIRLS) algorithm. Initialising the algorithm with $\muh_i=y_i$ and $\etah_i=g(\muh_i)$, we define so-called pseudodata as $z_i=g'(\muh_i)(y_i-\muh_i)+\etah_i$ and iterative weights $w_i=1/(g'(\muh_i)^2V(\muh_i))$ where $V(\mu_i)=\var(Y_i)\phi=b_i''(\theta)a_i(\phi)/\phi$ is the variance function for the distribution of $Y_i$. Let $\betah$ and $\fbh$ be the minimisers of
\begin{equation}\label{eqn:glm_min}
\Vert\sqrt{\WB}(\zb-\beta\xb-\fb)\Vert^2+n\phi\lambda\fb^T\boldsymbol{\Gamma}\fb
\end{equation}
where $\WB=\textrm{diag}(w_1,\ldots,w_n)$ is the weights matrix, $\zb=(z_1,\ldots,z_n)^T$, and $\lambda>0$ and $\boldsymbol{\Gamma}$ are as in (\ref{eqn:ps_min}). Now redefining $\etah_i=\betah\xb+\fbh$ and $\muh_i=g^{-1}(\etah_i)$, the algorithm is reiterated until convergence and the partial thin plate spline estimates $\betahps$ and $\fbhps$ are then the minimisers of (\ref{eqn:glm_min}) in the final iteration. Note that, if no smoothing is applied, $\betahps$ and $\fbhps$ are the maximum likelihood estimates in a generalized linear model (GLM), which are asymptotically unbiased.

\subsection{Spatial+ model}\label{sec:spatial_plus_glm}
Starting with the model (\ref{eqn:spatial_glm}), let $\WB$ and $\zb$ denote the weights matrix and pseudodata at convergence of the PIRLS algorithm. We then define the corresponding spatial+ model as follows. Let $\fbh^x$ and $\rb^x=\xb-\fbh^x=(r_1^x,\ldots,r_n^x)^T$ denote the fitted values and residuals in the weighted thin plate regression (\ref{eqn:co_model}) with weights $\WB$, i.e.\ $\fbh^x$ is the minimiser of $\Vert\sqrt{\WB}(\xb-\fb^x)\Vert^2+n\lambda_x\fb^{xT}\boldsymbol{\Gamma}\fb^x$ with smoothing parameter $\lambda_x>0$ and $\boldsymbol{\Gamma}$ defined as before. The spatial+ model is then the partial thin plate spline model defined by
\begin{equation}\label{eqn:spatial_plus_glm}
g(\mu_i)=\beta r_i^x+f^+(\tb_i)
\end{equation}
where $\beta$ and $f^+$ are estimated as described in Section \ref{sec:glm_spatial}. From Section \ref{sec:glm_spatial} we see that the estimates $\betahps$ and $\fbhps$ in the spatial model (\ref{eqn:spatial_glm}) are obtained as the minimisers of (\ref{eqn:ps_min}) if we replace $\yb,\xb,\fb,\boldsymbol{\Gamma}$ and $\lambda$ by $\tilde{\yb}=\sqrt{\WB}\zb$, $\tilde{\xb}=\sqrt{\WB}\xb$, $\tilde{\fb}=\sqrt{\WB}\fb$, $\tilde{\boldsymbol{\Gamma}}=\sqrt{\WB}^{-1}\boldsymbol{\Gamma}\sqrt{\WB}^{-1}$ and $\tilde{\lambda}=\phi\lambda$. Thus, at convergence of the PIRLS algorithm, estimation corresponds to that of a Gaussian model for which the model matrix has columns $\tilde{\xb}$ and $\sqrt{\WB}\BBs$. From our comment at the end of Section \ref{sec:sim_results}, the decorrelation trick that we used in Section \ref{sec:ts} would therefore work if we replace $\tilde{\xb}$ by $\tilde{\rb}$, obtained from a decomposition $\tilde{\xb}=\tilde{\vb}+\tilde{\rb}$ in which $\tilde{\vb}$ is in the column space of $\sqrt{\WB}\BBs$ and $\tilde{\rb}$ is broadly orthogonal to the columns of $\sqrt{\WB}\BBs$. By the properties of weighted thin plate spline regressions, $\sqrt{\WB}\rb^x$ is broadly orthogonal to $\sqrt{\WB}\BBs$. Therefore, letting $\tilde{\vb}=\sqrt{\WB}\fbh^x$ and $\tilde{\rb}=\sqrt{\WB}\rb^x$, the required decorrelation is achieved. Finally, replacing $\tilde{\xb}$ by $\tilde{\rb}$ is equivalent to replacing $\xb$ by $\rb^x$ in the spatial model, leading to the model (\ref{eqn:spatial_plus_glm}).

\subsection{Simulations}
The models (\ref{eqn:spatial_glm}) and (\ref{eqn:spatial_plus_glm}) can once again be implemented using thin plate regression splines in \texttt{mgcv}.
To test the performance of the spatial+ model (\ref{eqn:spatial_plus_glm}), we repeat the simulations from Section \ref{sec:sim} for three different response distributions, namely, the Poisson distribution with canonical link function $g(\mu)=\log(\mu)$, the exponential distribution with (non-canonical) link function $g(\mu)=\log(\mu)$ and the binomial distribution with size parameter $n_{\textrm{bin}}=10$ and canonical link function $g(\mu)=\log(\mu/(n_{\textrm{bin}}-\mu))$. 

For each response distribution, we simulate 100 replicates of the response data $\yb=(y_1,\ldots,y_n)^T$ by independently sampling each $y_i$ from the given distribution with mean $\mu_i=g^{-1}(\eta_i)$ where $\eta_i=\beta x_i+f_i$ with $\xb=(x_1,\ldots,x_n)^T$ simulated as in Section \ref{sec:sim_data},  $\sigma_x=0.1$, and true effects $\beta=3$, $\fb=(f_1,\ldots,f_n)^T=-\zb-\zb'$ as before. The results of fitting the models (\ref{eqn:spatial_glm}) and (\ref{eqn:spatial_plus_glm}) are summarised in Figure \ref{fig:sim_results_glm}. For comparison we have also included the results of fitting the corresponding null model (i.e.\ the GLM defined by $g(\mu_i)=\beta x_i$) and the models (\ref{eqn:spatial_glm}) and (\ref{eqn:spatial_plus_glm}) with no smoothing penalty applied. Finally, we have fitted a generalised version of the RSR model (for details see Appendix E). Note that we have not included the gSEM here as it is not immediately clear how to generalise this model to non-Gaussian response distributions.

We see that for all three response distributions, the overall behaviour of the models is similar to what we saw in the Gaussian case. As before, the null model and RSR model both have highly biased covariate effect estimates, however, note that unlike the Gaussian case, the estimate is not the same in the two models. This is because, while in both models the estimate is given by $\betah=(\xb^T\WB\xb)^{-1}\xb^T\WB\zb$, the fitted values, and hence the weights and pseudodata at convergence, differ. Without smoothing, as expected, the spatial and spatial+ models give the same results, however, while the covariate effect estimate looks unbiased for the Poisson and exponential response distributions, it looks slightly biased for the binomial distribution, though not materially. This is not surprising as GLMs are only asymptotically unbiased and may have some bias in practice, particularly, when the number of estimated parameters is relatively large as it is in this case \citep{cox1968}. When smoothing is applied, MSE reduces as intended, but the covariate effect estimate in the spatial model becomes significantly biased while it remains broadly unbiased in the spatial+ model.

\begin{figure}
\begin{tabular}{cc}
\includegraphics[width=0.45\textwidth]{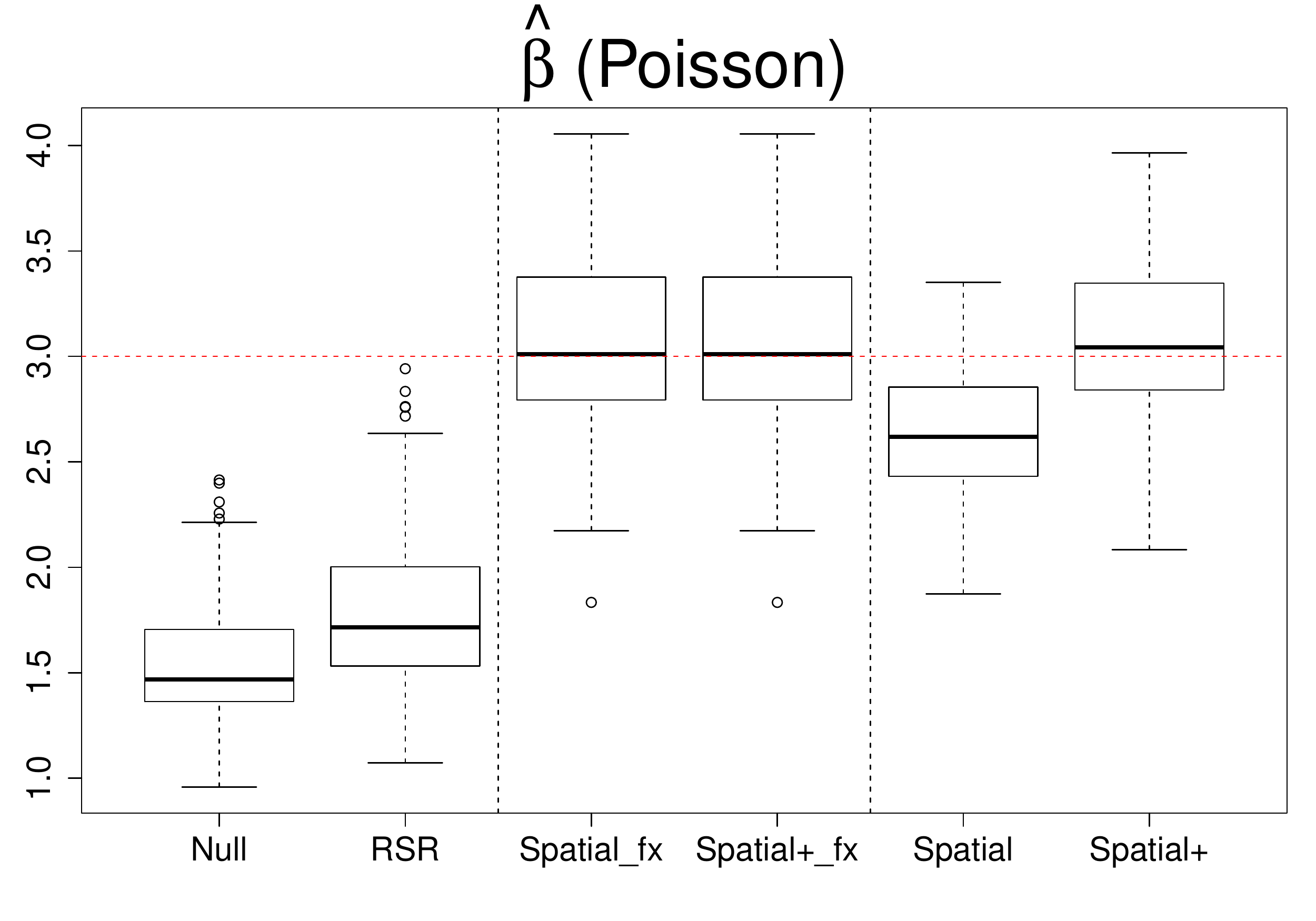}&
\includegraphics[width=0.45\textwidth]{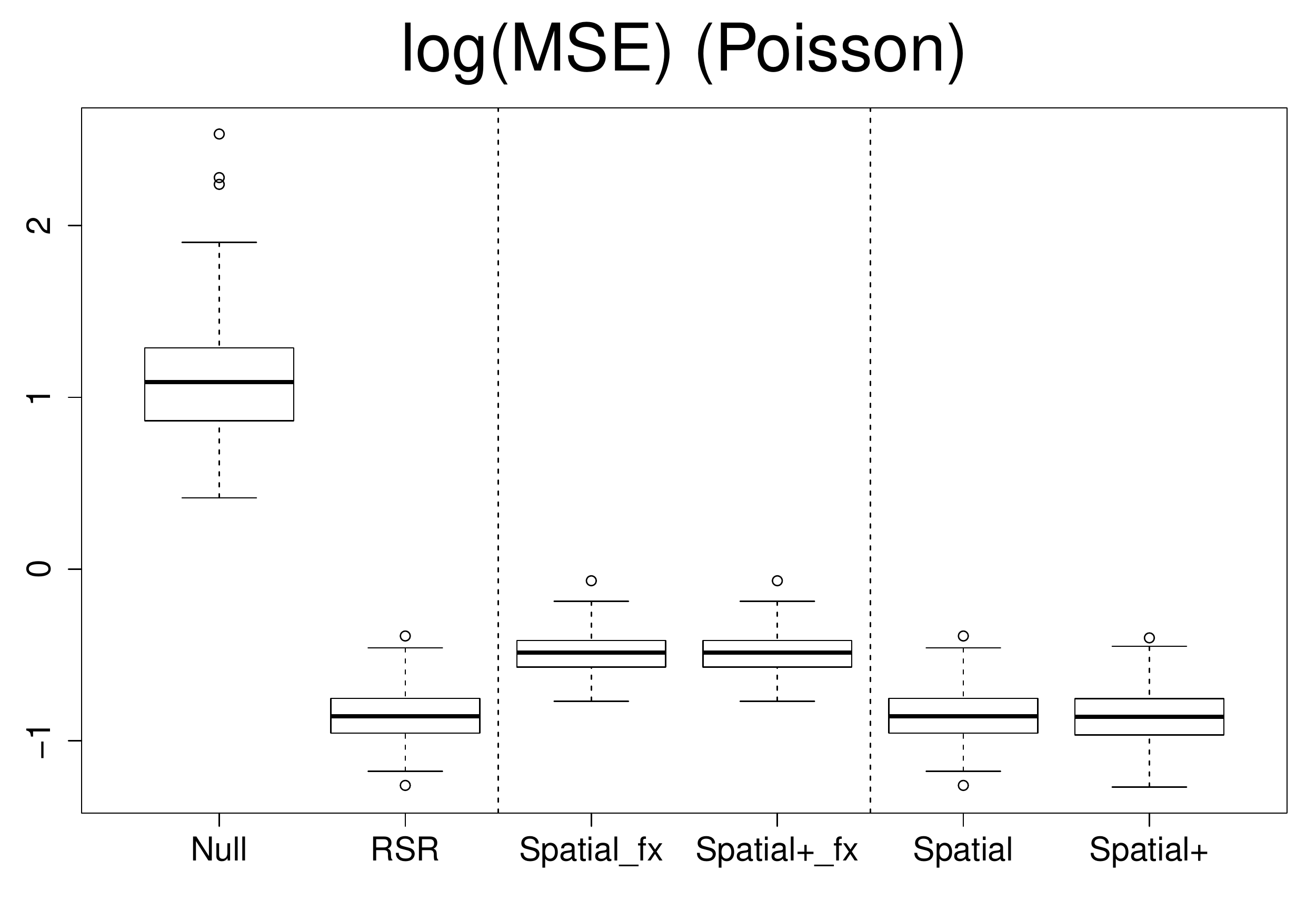}\\
\includegraphics[width=0.45\textwidth]{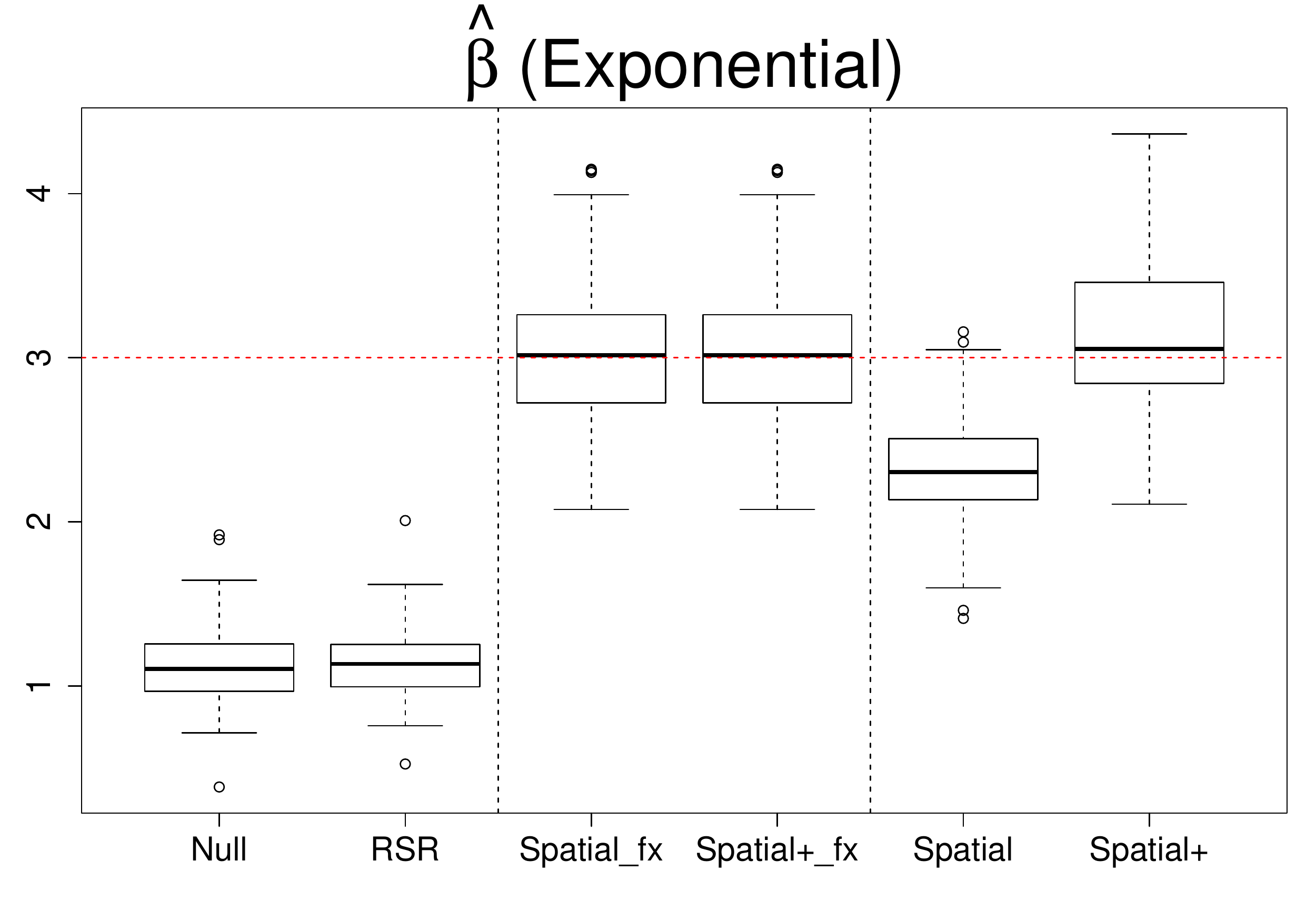}&
\includegraphics[width=0.45\textwidth]{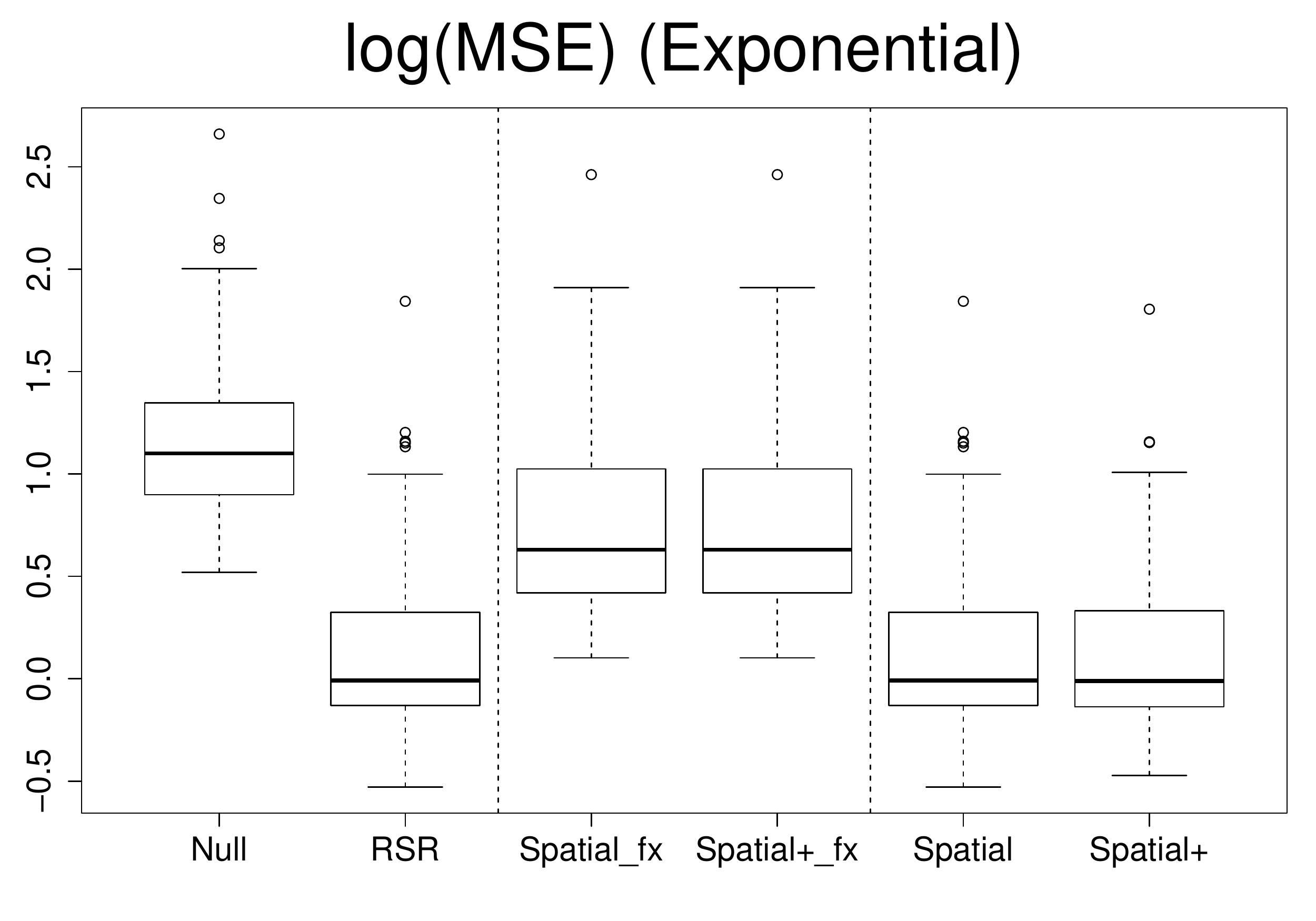}\\
\includegraphics[width=0.45\textwidth]{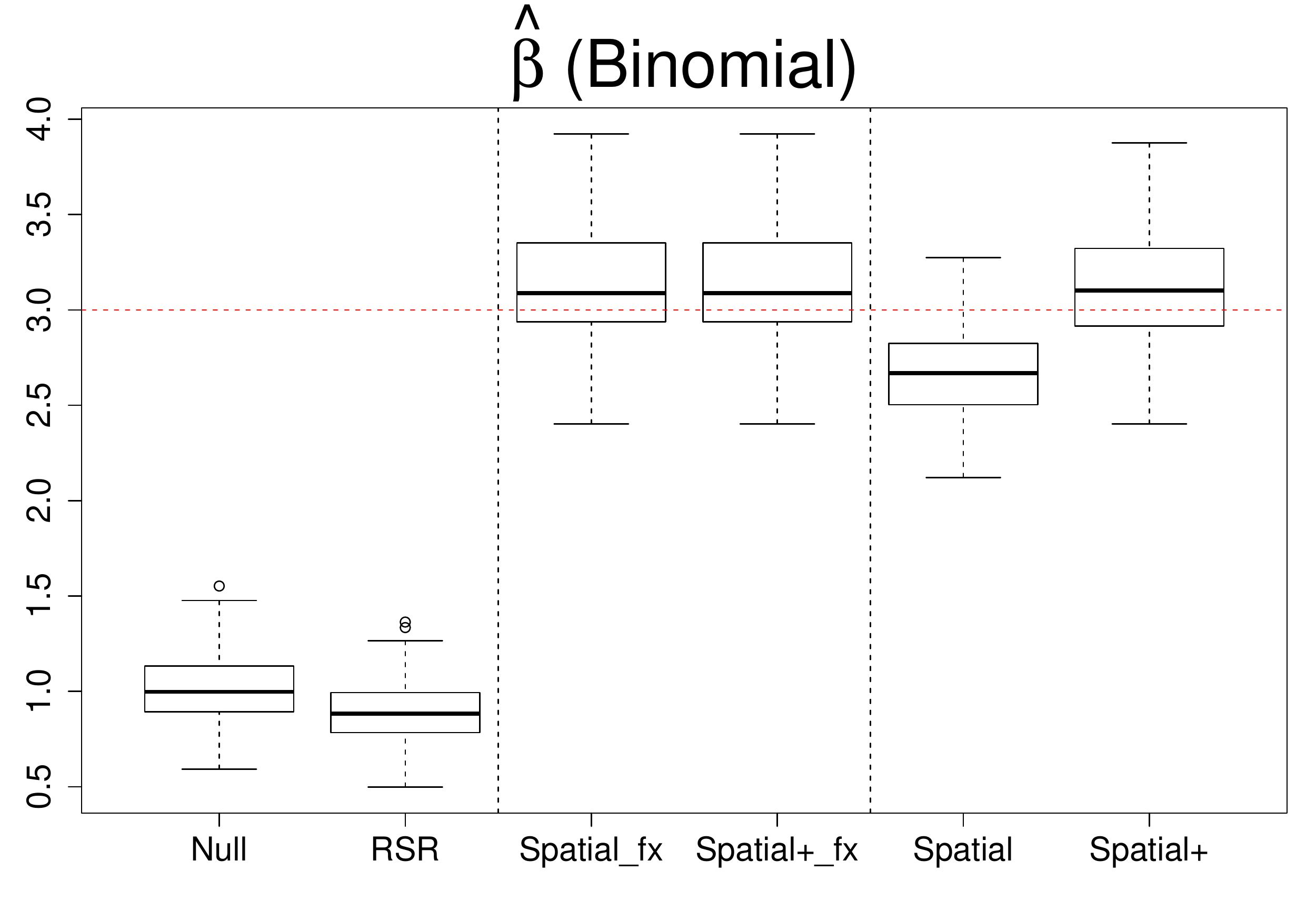}&
\includegraphics[width=0.45\textwidth]{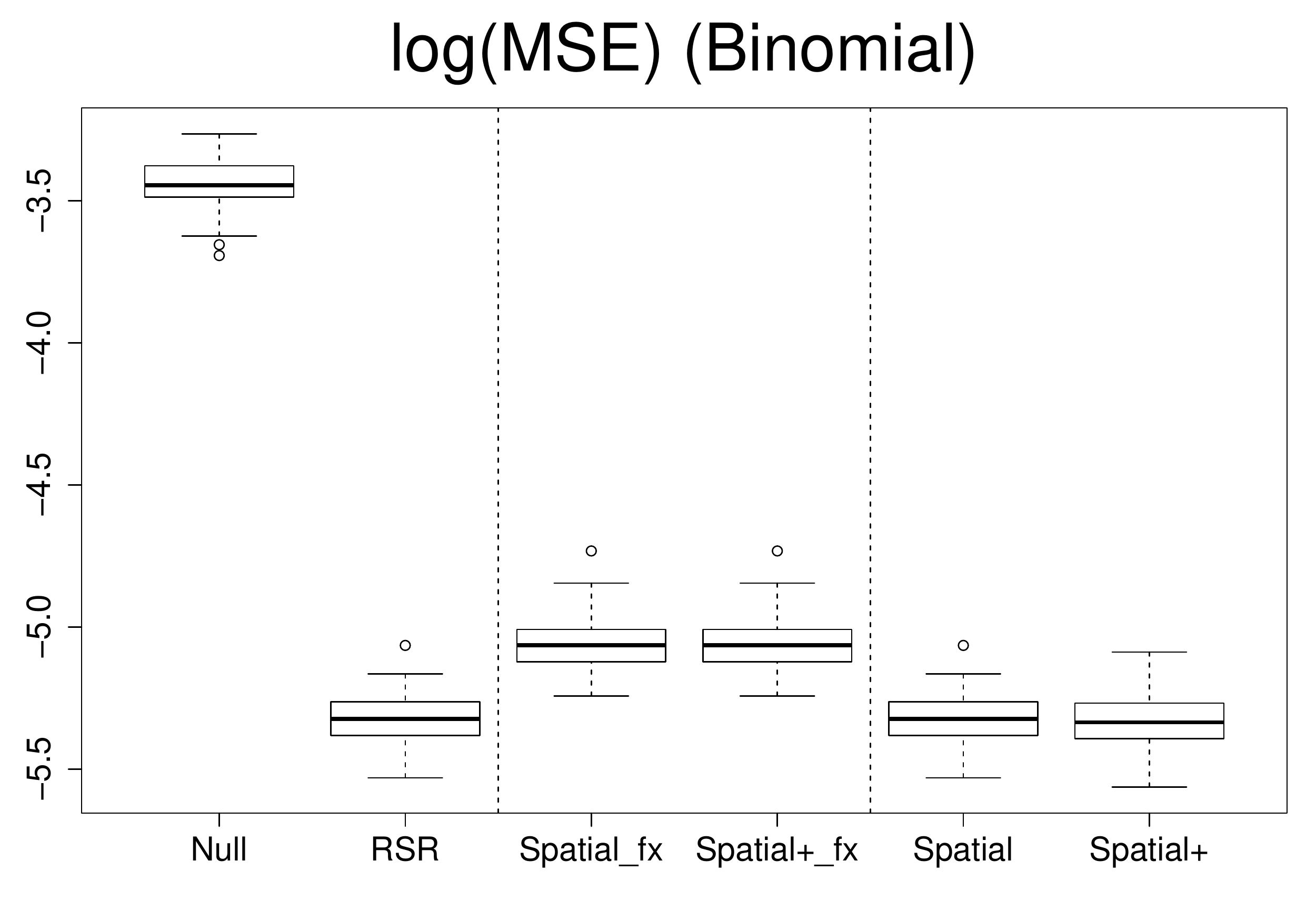}\\
\end{tabular}
\bigskip
\caption{For each of the distributions Poisson (top), exponential (middle), binomial (bottom): the estimated covariate effect $\hat{\beta}$ (left) and log(MSE) of fitted values (right) for each model fitted to 100 data replicates, where the true covariate effect is $\beta=3$. Subscript \texttt{fx} refers to a model in which no smoothing penalties were applied.}\label{fig:sim_results_glm}
\end{figure}

\section{Discussion}
We have shown that the proposed spatial+ model can be used to avoid unreliable covariate effect estimates in spatial regression with clear advantages over existing methods.
Our analysis also gives a clearer understanding of why spatial confounding happens. Spatial models, whether formulated in terms of spatially induced prior distributions or smoothing penalties, usually apply some form of spatial smoothing to reflect spatial correlation in the data and avoid overfitting. However, from the model formulation (\ref{eqn:spatial}), we see that it is exactly this smoothing that causes spatial confounding bias. If spatial location is highly explanatory for a covariate in the model, the model matrix has nearly collinear columns, and therefore the smoothing penalty can heavily influence the way in which the "total covariate spatial effect" (i.e.\ the observed effect of the covariate spatial pattern, including any unmeasured effects) is split between the covariate term and the spatial term. As we have seen, this can lead to significantly biased results.

The excessive smoothing-induced bias is avoided in both spatial+ and the gSEM. If no smoothing penalty is applied, both models give the same unbiased covariate effect estimates as the unsmoothed spatial model. 
Spatial+ reparametrises the spatial model so that, rather than splitting the total covariate spatial effect into two separate terms, it is fully contained in the term $f^+$. This makes fixed effect estimates broadly independent of the spatial effects, in particular, they remain largely unbiased under spatial smoothing. The idea of decorrelating covariate and spatial terms is also used in RSR. But in RSR, as it is achieved by restricting the spatial effects, it is the covariate effect estimates that contain the total covariate spatial effects, leading to bias by construction. 
In the gSEM, the elimination of all spatial information means that fixed effect estimates are once again decorrelated from the spatial effects and thereby protected from spatial smoothing. The resulting model of residuals only, however, seems less intuitive than spatial+ and, the change in response variable means that standard model selection criteria cannot be used for comparisons with the other models. A major advantage of spatial+ is also that the method generalises easily to models with non-Gaussian response distributions and our simulations illustrate that the method still works well here.

Our above discussion shows that the decorrelation of effect estimates is the underlying reason why the spatial+ approach works. As mentioned in Section \ref{sec:intro}, the modification of the model matrix that achieves this is easily transferable to other spatial model formulations, and we would therefore expect the method to work well in general. However, as our theoretical derivations are specific to thin plate spline estimates, similar derivations or simulations could be done to confirm our results in other settings. One limitation to the spatial+ approach is that the covariate effects in the model must be linear. This assumption is needed for the spatial residuals to capture the true covariate effects.
The spatial model (\ref{eqn:spatial}) is easily extended, using the generalized additive model (GAM) framework, to include non-linear covariate terms in the form of smooths (i.e.\ unknown functions of the covariates estimated from the data). It would be interesting to see if any of the ideas of spatial+, as well as our increased understanding of spatial confounding, can be used to develop methods for avoiding spatial confounding in this context.

Finally, applying spatial+ to the forestry example, we see that the effect of temperature on crown defoliation appears to be positive and significant as expected, and that this effect would likely be underestimated in both size and significance in the spatial model (and even more so in the null model). The other covariate, age of trees, in this example also illustrates that, if a covariate is not spatially confounded, this can be confirmed by showing that its effect estimate in the spatial and spatial+ models agree. It is possible that this idea could be used to develop a diagnostic or test that practitioners could use to identify spatial confounding in applications.

\section*{Acknowledgements}
Emiko Dupont is supported by a scholarship from the EPSRC Centre for Doctoral Training in Statistical Applied Mathematics at Bath (SAMBa), under the project EP/L015684/1. We thank The Forest Institute Baden-W\"urttemberg (Germany) for making the 2013 Terrestrial Crown
Condition Inventory (TCCI) forest health monitoring survey data available.

\appendix
\numberwithin{equation}{section}
\section*{Appendices}
Appendices referenced in Sections \ref{sec:intro}, \ref{sec:asym}, \ref{sec:sim} and \ref{sec:glm} of the paper 
are included below.

\section{Technical lemmas}\label{app:asym_pre}
In this appendix we set out the technical lemmas that we use for the derivations of the main results in Sections $\secasymps$ and $\secasymts$ of the paper which generalise the results of \citet{rice1986,chen_shiau1991} from $d=1$ to dimensions $d\ge 1$. Key to this generalisation is the following result by \citet{utreras1988} on the asymptotics of thin plate splines. 

\begin{lemma}\label{lem:eigen}
Suppose $\Omega$ has Lipschitz boundary and satisfies a uniform cone condition (as defined in \cite{utreras1988}). Assume that the points $\{\tb_1,\ldots,\tb_n\}\subset\Omega$ are regularly distributed in the sense that there exists a constant $B>0$ such that
\[
\frac{h_{\min}}{h_{\max}}\le B
\]
where $h_{\max}=\sup_{\tb\in\Omega}\inf_{i}\vert \tb-\tb_i\vert$ and $h_{\min}=\min_{i\ne j}\vert \tb_i-\tb_j\vert$. 
Let $\mu_1\le\cdots\le\mu_n$ denote the eigenvalues of the matrix $n\boldsymbol{\Gamma}$ and assume $m>d/2$. Then 
\[
\mu_1=\cdots=\mu_M=0
\]
and there exist constants $C_1,C_2>0$ such that
\[
C_1k^{2m/d}\le\mu_k\le C_2 k^{2m/d} \quad \text{for }M+1\le k\le n.
\]
\end{lemma}
\begin{proof}
See the proof of Theorem 5.1 (a) and Theorem 5.3 of \citet{utreras1988}.
\end{proof}
Lemma \ref{lem:eigen} provides us with a convenient basis in which the smoother matrix $\SBl=(\IB+n\lambda\boldsymbol{\Gamma})^{-1}$ 
is diagonalised and, moreover, describes the asymptotic behaviour of its eigenvalues as the number of data points $n\rightarrow \infty$.
More specifically, if $\PPhi$ is the matrix whose columns are $\frac{1}{\sqrt{n}}\pphi_1,\ldots,\frac{1}{\sqrt{n}}\pphi_n$ where $\pphi_k$ is an eigenvector of $n\Gamma$ corresponding to the eigenvalue $\mu_k$, then (with appropriate scaling of the eigenvectors) $\PPhi$ has orthonormal columns and
\begin{eqnarray*}
\PPhi^T\SBl\PPhi&=&\text{diag}\big(1/(1+\lambda\mu_1),\ldots,1/(1+\lambda\mu_n)\big),\\
\PPhi^T(\IB-\SBl)\PPhi&=&\text{diag}\big((\lambda\mu_1)/(1+\lambda\mu_1),\ldots,(\lambda\mu_n)/(1+\lambda\mu_n)\big).
\end{eqnarray*}
This representation allows us to explicitly evaluate the estimates in the models of dimension $d\ge 1$ which, in turn, enables us to obtain asymptotic results in a similar way to \citet{rice1986, chen_shiau1991}. 

For the rest of these appendices, we assume that $m> d/2$ and that the domain $\Omega$ and the data points $\tb_1,\ldots,\tb_n$ satisfy the conditions of Lemma \ref{lem:eigen}. We will also use the notation $a(n)\approx b(n)$ to mean that $a(n)/b(n)$ is bounded away from zero and infinity as $n\rightarrow \infty$.

Lemmas \ref{lem:trace} and \ref{lem:sqbias} link the asymptotic behaviour of the smoother matrix $\SBl$ to the convergence rate of the smoothing parameter $\lambda$. 
Lemma \ref{lem:trace} generalises Lemma 2 of \citet{chen_shiau1991} to dimensions $d\ge 1$, and is proved using the asymptotic properties of the eigenvalues given in Lemma \ref{lem:eigen}.
The result in Lemma \ref{lem:sqbias} is proved by \citet{utreras1988}.
Lemmas \ref{lem:tech_rice} and \ref{lem:tech} prove a number of asymptotic results that are convenient for later proofs. Lemma \ref{lem:tech_rice} shows how the results used by \citet{rice1986} for the analysis in dimension $d=1$ generalise to dimensions $d\ge 1$, while Lemma \ref{lem:tech} generalises Lemma 3 of \citet{chen_shiau1991} to dimensions $d\ge 1$. Proofs of Lemmas \ref{lem:trace}, \ref{lem:tech_rice} and \ref{lem:tech} are given in Appendix B.

\begin{lemma}\label{lem:trace}
Suppose $\lambda\approx n^{-\delta}$ for some $0<\delta<1$. Then
\begin{description}
\item[(a)] $\text{Tr}(\SBl)= \sum_{k=1}^n (1+\lambda\mu_k)^{-1}=M+\LO(\lambda^{-d/2m})$,
\item[(b)] $\text{Tr}(\SBl^2)=\sum_{k=1}^n (1+\lambda\mu_k)^{-2}=M+\LO(\lambda^{-d/2m})$.
\end{description}
In particular, if $m\ge d$, then both of these sums are of the form $\LO(n^{1/2-\tau})$ where $0<\tau<1/2$ depends only on $\delta$.
\end{lemma}
\begin{proof}
See Appendix B.
\end{proof}

\begin{lemma}\label{lem:sqbias}
For any $g\in H^m(\Omega)$, let $\gb=(g(\tb_1),\ldots,g(\tb_n))^T$. The averaged squared bias $\Btp(g,\lambda)$ of the thin plate spline $\SBl\gb$ (i.e.\ the fitted values in a model of the form (\eqnspatial) in our paper with $\beta=0$) is given by
\[
\Btp(g,\lambda)=\frac{1}{n}\gb^T(\IB-\SBl)^2\gb=\LO(\lambda).
\]
\end{lemma}
\begin{proof}
See \citet{utreras1988} Lemma 2.2.
\end{proof}

\begin{lemma}\label{lem:tech_rice}
Suppose $\lambda\approx n^{-\delta}$ for some $0<\delta<1$, $f, f^x\in H^m(\Omega)$ are bounded and $m\ge d$. Let $\fb=(f(\tb_1),\ldots,f(\tb_n))^T$. Then
\begin{description}
\item[(a)] $n^{-1}\xb^T(\IB-\SBl)\xb=\sigma_x^2+\lo(1)$,
\item[(b)] $n^{-1}\xb^T(\IB-\SBl)^2\xb=\sigma_x^2+\lo(1)$,
\item[(c)] $n^{-1}\xb^T(\IB-\SBl)\fb=\lo(n^{-1/2})+\LO(\lambda^{-1/2})$,
\item[(d)] $n^{-1}\xb^T\SBl^2\xb=\LO(1)$
\end{description}
\end{lemma}
\begin{proof}
See Appendix B.
\end{proof}

\begin{lemma}\label{lem:tech}
Suppose $\lambda\approx n^{-\delta}, \lambda_x\approx n^{-\delta_x}$ for some $0<\delta,\delta_x< 1$, $f,f^x\in H^m(\Omega)$ and $m\ge d$. Let $\fb=(f(\tb_1),\ldots,f(\tb_n))^T$. Then
\begin{description}
\item[(a)] $n^{-1}\xb^T(\IB-\SBlx)(\IB-\SBl)(\IB-\SBlx)\xb=\sigma_x^2+\lo(1)$,
\item[(b)] $n^{-1}\xb^T(\IB-\SBlx)(\IB-\SBl)^2(\IB-\SBlx)\xb=\sigma_x^2+\lo(1)$,
\item[(c)] $n^{-1}\xb^T(\IB-\SBlx)(\IB-\SBl)\fb=\lo(n^{-1/2})+\LO((\lambda_x\lambda)^{1/2})$,
\item[(d)] $n^{-1}\xb^T(\IB-\SBlx)(\IB-\SBl)\SBlx\xb=\lo(n^{-1/2})+\LO((\lambda_x\lambda)^{1/2})$,
\item[(e)] $n^{-1}\xb^T\SBlx(\IB-\SBl)^2\SBlx\xb=\LO(\lambda)+\LO(n^{-1}\lambda_x^{-d/2m}\log^2n )$
\item[(f)] $n^{-1}\xb^T[\SBl+(\IB-\SBl)\SBlx]^T[\SBl+(\IB-\SBl)\SBlx]\xb=\LO(1)$,
\end{description}
\end{lemma}
\begin{proof}
See Appendix B.
\end{proof}

\section{Proofs of technical lemmas}\label{app:proofs_lem}
In this appendix we prove the lemmas set out in Appendix A. We start by introducing some notation. Recall the assumption from our paper 
that
\[
x_i=f^x(\tb_i)+\epsilon^x_i,\quad \epsilon^x_i\sim_{\text{iid}} N(0,\sigma_x^2)
\]
which means that the covariate $\xb$ is correlated with the smooth $f$ in the spatial model. 
Therefore, $\xb$ decomposes as 
\begin{equation}\label{eqn:fx_decomp}
\xb=\fb^x+\eepsilon^x
\end{equation} 
with $\fb^x=(f^x(\tb_1),\ldots,f^x(\tb_n))^T$ and $\eepsilon^x=(\epsilon^x_1,\ldots\epsilon^x_n)^T$. For the asymptotic analysis, it is often convenient to consider the behaviour of the components in this decomposition separately. Let $\cb^x=(c^x_1,\ldots,c^x_n)^T$ and $\xxi^x=(\xi^x_1,\ldots,\xi^x_n)^T$ denote the coefficients of $\fb^x$ and $\eepsilon^x$, respectively, in the basis $\PPhi$ introduced in Appendix A, i.e.\ 
\begin{eqnarray*}
\fb^x&=&\PPhi\cb^x\quad\text{ where }\cb^x=\PPhi^T\fb^x,\\
\eepsilon^x&=&\PPhi\xxi^x\quad\text{ where }\xxi^x=\PPhi^T\eepsilon^x.
\end{eqnarray*}
Note that since $f^x\in H^m(\Omega)$ is bounded, we have that
\begin{equation}\label{eqn:fx_bounded}
n^{-1}\sum_{k=1}^n(c^x_k)^2=n^{-1}(\fb^x)^T(\fb^x)\rightarrow 0 \text{ as } n\rightarrow\infty.
\end{equation}
As in \cite{rice1986} and \cite{chen_shiau1991}, we also note that the following assumptions hold for the coefficients $\xxi^x$ of the iid noise $\eepsilon^x$.
\begin{description}
\item[(A1)] $n^{-1}\sum_{k=1}^n\xi^x_k\rightarrow 0$ as $n\rightarrow\infty$,
\item[(A2)] $n^{-1}\sum_{k=1}^n(\xi^x_k)^2=n^{-1}(\eepsilon^x)^T\eepsilon^x\rightarrow \sigma_x^2>0$ as $n\rightarrow\infty$,
\item[(A3)] $\sup_{1\le k\le n}\vert \xi^x_k\vert = \LO(\log n)$.
\end{description}

\subsection*{Proof of Lemma \ref{lem:trace}}
From Lemma \ref{lem:eigen}, $\mu_k=0$ for $k=1,\ldots,M$, so $\sum_{k=1}^M (1+\lambda\mu_k)^{-1}=M$. Split the remaining range of the summation into $I_1=[M+1,\lambda^{-d/2m}]$, $I_2=[\lambda^{-d/2m},n]$. 

$I_1$: Since $(1+\lambda\mu_k)^{-1}\le 1$ for all $k$
\[
\sum_{I_1}(1+\lambda\mu_k)^{-1}\le \sum_{I_1}1\le\lambda^{-d/2m}.
\]

$I_2$: By Lemma \ref{lem:eigen}, $(1+\lambda\mu_k)^{-1}\le (C_1\lambda  k^{2m/d})^{-1}$ for all $k$ in $I_2$. Since $\{\mu_k\}_k$ is an increasing sequence, we have that
\begin{eqnarray*}
\sum_{I_2}(1+\lambda\mu_k)^{-1}&\le&\int_{\lambda^{-d/2m}}^{\infty}(C_1 \lambda x^{2m/d})^{-1} dx\\
&=&C\lambda^{-d/2m}
\end{eqnarray*}
where $C= (C_1 (2m/d-1))^{-1}$. This proves part (a).




%

For part (b) we note that $\sum_{k=1}^M (1+\lambda\mu_k)^{-2}=M$ as before and that $(1+\lambda\mu_k)^{-2} <(1+\lambda\mu_k)^{-1}$ for all the remaining $k$. Therefore (b) follows from (a).

\subsection*{Proof of Lemma \ref{lem:tech_rice}}
To prove (a), we use the decomposition $\xb=\fb^x+\eepsilon^x$ from (\ref{eqn:fx_decomp}) and the corresponding basis expansions in the basis $\PPhi$ to get
\[
n^{-1}\xb^T(\IB-\SBl)\xb=n^{-1}\sum_k(c^x_k+\xi^x_k)^2\frac{\lambda\mu_k}{1+\lambda\mu_k}.
\]
We note that while 
\[
(c^x_k+\xi^x_k)^2=(c^x_k)^2+(\xi^x_k)^2+2c^x_k\xi^x_k,
\]
due to the Cauchy-Schwarz inequality, the term $2c^x_k\xi^x_k$ will never dominate the rate of convergence. Therefore, we only need to consider the parts of the sum relating to the other two terms.
Using Cauchy-Schwarz again we see that
\begin{eqnarray*}
\sum_k(c^x_k)^2\frac{\lambda\mu_k}{1+\lambda\mu_k}
&\le&\bigg(\sum_k(c^x_k)^2\big(\frac{\lambda\mu_k}{1+\lambda\mu_k}\big)^2\bigg)^{1/2}\bigg(\sum_k(c^x_k)^2\bigg)^{1/2}\\
&=&\big(n\Btp(f^x,\lambda)\big)^{1/2}\big(\sum_k(c^x_k)^2\big)^{1/2}\\
&=&\LO(n\lambda^{1/2})=\LO(n^{1-\delta/2})=\lo(n).
\end{eqnarray*}
Here we have used Lemma \ref{lem:sqbias} and (\ref{eqn:fx_bounded}).

For the term involving $(\xi^x_k)^2$ we have that
\begin{eqnarray*}
\sum_k(\xi^x_k)^2-\sum_k(\xi^x_k)^2\frac{\lambda\mu_k}{1+\lambda\mu_k}&=&\sum_k(\xi^x_k)^2\frac{1}{1+\lambda\mu_k}\\
&\le&\sup_k(\xi^x_k)^2\sum_k\frac{1}{1+\lambda\mu_k}\\
&=&\LO(\log^2n)\LO( n^{1/2-\tau})=\lo(n)
\end{eqnarray*}
by assumption (A3) and Lemma \ref{lem:trace}.
Hence, by assumption (A2),
\[
n^{-1}\sum_k(\xi^x_k)^2\frac{\lambda\mu_k}{1+\lambda\mu_k}\rightarrow \sigma_x^2\quad \text{as }n\rightarrow \infty,
\]
and therefore (a) is proved.

For (b) we write
\[
n^{-1}\xb^T(\IB-\SBl)^2\xb=n^{-1}\sum_k(c^x_k+\xi^x_k)^2\bigg(\frac{\lambda\mu_k}{1+\lambda\mu_k}\bigg)^2.
\]
By Lemma \ref{lem:sqbias} we have that
\[
n^{-1}\sum_k(c^x_k)^2\bigg(\frac{\lambda\mu_k}{1+\lambda\mu_k}\bigg)^2=\Btp(f^x,\lambda)=\LO(\lambda)=\lo(1).
\]
For $a>0$ we have $\frac{1}{1+a}\le1$ and $\frac{a}{1+a}\le1$ and therefore
\[
1-\bigg(\frac{a}{1+a}\bigg)^2=\frac{(1+a)^2-a^2}{(1+a)^2}=\frac{(1+a)+a}{(1+a)^2}\le\frac{2}{1+a}.
\]
Using this with $a=\lambda\mu_k$ we see from assumption (A3) and Lemma \ref{lem:trace} that
\begin{eqnarray*}
\sum_k(\xi^x_k)^2-\sum_k(\xi^x_k)^2\bigg(\frac{\lambda\mu_k}{1+\lambda\mu_k}\bigg)^2&\le& \sup_k(\xi^x_k)^2\sum_k\frac{2}{1+\lambda\mu_k}\\
&=&\LO\big((\log^2n) n^{1/2-\tau}\big)=\lo(n).
\end{eqnarray*}
So by assumption (A2), (b) is proved.

For (c) let $\cb=\PPhi^T\fb$ be the coefficients of $\fb$ in the basis $\PPhi$. Then
\[
n^{-1}\xb^T(\IB-\SBl)\fb=n^{-1}\sum_k (c^x_kc_k+\xi^x_kc_k)\frac{\lambda\mu_k}{1+\lambda\mu_k}.
\]
For the term involving $c^x_k$, we use Cauchy-Schwarz and \eqref{eqn:fx_bounded} to see that
\begin{eqnarray*}
\bigg\vert n^{-1}\sum_k c^x_kc_k\frac{\lambda\mu_k}{1+\lambda\mu_k}\bigg\vert&\le& \bigg(n^{-1}\sum_k(c^x_k)^2\bigg)^{1/2}\bigg(n^{-1}\sum_k\bigg(\frac{c_k\lambda\mu_k}{1+\lambda\mu_k}\bigg)^2\bigg)^{1/2}\\
&=&\LO\big((\Btp(f,\lambda)\big)^{1/2})=\LO(\lambda^{1/2})
\end{eqnarray*}
by Lemma \ref{lem:sqbias}. For the term involving $\xi^x_k$, we use  Cauchy-Schwarz again to obtain
\begin{eqnarray*}
\bigg\vert n^{-1}\sum_k \xi^x_kc_k\frac{\lambda\mu_k}{1+\lambda\mu_k}\bigg\vert&\le&\lambda^{1/2}\sup_k\vert \xi^x_k\vert \bigg\vert n^{-1}\sum_k c_k\mu_k^{1/2}\frac{(\lambda\mu_k)^{1/2}}{1+\lambda\mu_k}\bigg\vert\\
&\le&\lambda^{1/2}\sup_k\vert \xi^x_k\vert\bigg(n^{-1}\sum_k c_k^2\mu_k\bigg)^{1/2}\bigg(n^{-1}\sum_k\frac{\lambda\mu_k}{(1+\lambda\mu_k)^2}\bigg)^{1/2}\\
&\le& \LO(\lambda^{1/2}\log n)\LO(n^{-1/2}\lambda^{-d/4m})=\lo(n^{-1/2})
\end{eqnarray*}
Here we have used assumption (A3), Lemma \ref{lem:trace} (since $\frac{\lambda\mu_k}{(1+\lambda\mu_k)^2}\le \frac{1}{1+\lambda\mu_k}$) and the fact that
\[
n^{-1}\sum_k c_k^2\mu_k=\fb^T \Gamma\fb \le \vert f\vert_m^2 < \infty
\]
since $f\in H^m(\Omega)$. 
The rate of convergence of $\lo(n^{-1/2})$ follows from the fact that 
\[
n^{-1/2}(\log n) \lambda^{-d/4m+1/2}\approx n^{-1/2}(\log n) n^{-\delta(1-d/2m)/2}=\lo(n^{-1/2})
\]
since $1-d/2m>0$. This proves (c).

For (d) we have that
\[
n^{-1}\xb^T\SBl^2\xb=n^{1}\sum_k(c^x_k+\xi^x_k)^2\frac{1}{(1+\lambda_x\mu_k)^2}.
\]
For the term involving $(c^x_k)^2$ we see that
\[
n^{-1}\sum_k(c^x_k)^2\frac{1}{(1+\lambda_x\mu_k)^2}\le n^{-1}\sum_k(c^x_k)^2=\LO(1)
\]
by \eqref{eqn:fx_bounded}. 
For the term involving $(\xi^x_k)^2$ we see from assumption (A3) and Lemma \ref{lem:trace} that
\begin{eqnarray*}
n^{-1}\sum_k(\xi^x_k)^2\frac{1}{(1+\lambda_x\mu_k)^2}&\le& n^{-1}\sup_k(\xi^x_k)^2\sum_k\frac{1}{(1+\lambda_x\mu_k)^2}\\
&=&\LO\big((\log^2n)n^{-1/2-\tau}\big)=\LO(1).
\end{eqnarray*}
Hence $n^{-1}\xb^T\SBl^2\xb=\LO(1)$.

\subsection*{Proof of Lemma \ref{lem:tech}}
As in the proof of Lemma \ref{lem:tech_rice} we write
\[
n^{-1}\xb^T(\IB-\SBlx)(\IB-\SBl)(\IB-\SBlx)\xb=n^{-1}\sum_k(c^x_k+\xi^x_k)^2\bigg(\frac{\lambda_x\mu_k}{1+\lambda_x\mu_k}\bigg)^2\frac{\lambda\mu_k}{1+\lambda\mu_k}
\]
and once again, by Cauchy-Schwarz, we only need to consider the terms involving $(c^x_k)^2$ and $(\xi^x_k)^2$. Since $\frac{\lambda\mu_k}{1+\lambda\mu_k}\le 1$, Lemma \ref{lem:sqbias} shows that
\[
n^{-1}\sum_k(c^x_k)^2\bigg(\frac{\lambda_x\mu_k}{1+\lambda_x\mu_k}\bigg)^2\frac{\lambda\mu_k}{1+\lambda\mu_k}\le n^{-1}(\fb^x)^T(\IB-\SBlx)^2\fb^x = \Btp(f^x,\lambda_x)=\LO(\lambda_x)=\lo(1).
\]
For the term involving $(\xi^x_k)^2$, firstly note that if $a_1,a_2,a_3>0$, then
\begin{eqnarray*}
1-\frac{a_1a_2a_3}{(1+a_1)(1+a_2)(1+a_3)}&=&\frac{(1+a_1)(1+a_2)(1+a_3)-a_1a_2a_3}{(1+a_1)(1+a_2)(1+a_3)}\\
&=&\frac{1+a_1+a_2+a_3+a_1a_2+a_1a_3+a_2a_3}{(1+a_1)(1+a_2)(1+a_3)}\\
&\le&\frac{3}{1+a_1}+\frac{2}{1+a_2}+\frac{2}{1+a_3}
\end{eqnarray*}
where in the last step we have used the fact that $\frac{1}{1+a_i}\le 1$ and $\frac{a_i}{1+a_i}\le 1$ for all $i$. Using this with $a_1=a_2=\lambda_x\mu_k$ and $a_3=\lambda\mu_k$ we see that
\begin{eqnarray*}
\sum_k(\xi^x_k)^2-\sum_k(\xi^x_k)^2\bigg(\frac{\lambda_x\mu_k}{1+\lambda_x\mu_k}\bigg)^2\frac{\lambda\mu_k}{1+\lambda\mu_k}
&\le&
\sup_k(\xi^x_k)^2\bigg(\sum_k\frac{5}{1+\lambda_x\mu_k}+\sum_k\frac{2}{1+\lambda_x\mu_k}\bigg)\\
&=&\LO(\log^2n)\LO(n^{1/2-\tau})=\lo(n)
\end{eqnarray*}
by assumption (A3) and Lemma \ref{lem:trace}. Therefore,
\[
n^{-1}\sum_k(\xi^x_k)^2\bigg(\frac{\lambda_x\mu_k}{1+\lambda_x\mu_k}\bigg)^2\frac{\lambda\mu_k}{1+\lambda\mu_k}\rightarrow \sigma_x^2
\]
by assumption (A2). This shows (a).

For (b) we have that
\[
n^{-1}\xb^T(\IB-\SBlx)(\IB-\SBl)^2(\IB-\SBlx)\xb=n^{-1}\sum_k(c^x_k+\xi^x_k)^2\bigg(\frac{\lambda_x\mu_k}{1+\lambda_x\mu_k}\bigg)^2\bigg(\frac{\lambda\mu_k}{1+\lambda\mu_k}\bigg)^2.
\]
For the term involving $(c^x_k)^2$, the same argument as in (a) shows that this is $\lo(1)$. For the $(\xi^x_k)^2$ term we note that
\[
1-\frac{a_1a_2a_3a_4}{(1+a_1)(1+a_2)(1+a_3)(1+a_4)}
\le\frac{5}{1+a_1}+\frac{4}{1+a_2}+\frac{4}{1+a_3}+\frac{2}{1+a_4}
\]
for $a_1,a_2,a_3,a_4>0$ and using this with $a_1=a_2=\lambda_x\mu_k$ and $a_3=a_4=\lambda\mu_k$ shows that
\[
n^{-1}\sum_k(\xi^x_k)^2\bigg(\frac{\lambda_x\mu_k}{1+\lambda_x\mu_k}\bigg)^2\bigg(\frac{\lambda\mu_k}{1+\lambda\mu_k}\bigg)^2\rightarrow \sigma_x^2
\]
as in (a). This proves (b).

For (c) let $\cb=\PPhi^T\fb$ be the coefficients of $\fb$ in the basis $\PPhi$. Then
\[
n^{-1}\xb^T(\IB-\SBlx)(\IB-\SBl)\fb=n^{-1}\sum_k (c^x_kc_k+\xi^x_kc_k)\frac{\lambda_x\mu_k}{1+\lambda_x\mu_k}\frac{\lambda\mu_k}{1+\lambda\mu_k}.
\]
For the term involving $c^x_k$, we use Cauchy-Schwarz to see that
\begin{eqnarray*}
\bigg\vert n^{-1}\sum_k c^x_kc_k\frac{\lambda_x\mu_k}{1+\lambda_x\mu_k}\frac{\lambda\mu_k}{1+\lambda\mu_k}\bigg\vert&\le& \bigg(n^{-1}\sum_k\bigg(\frac{c^x_k\lambda_x\mu_k}{1+\lambda_x\mu_k}\bigg)^2\bigg)^{1/2}\bigg(n^{-1}\sum_k\bigg(\frac{c_k\lambda\mu_k}{1+\lambda\mu_k}\bigg)^2\bigg)^{1/2}\\
&=&\bigg(\Btp(f^x,\lambda_x)\Btp(f,\lambda)\bigg)^{1/2}=\LO\big((\lambda_x\lambda)^{1/2}\big)
\end{eqnarray*}
by Lemma \ref{lem:sqbias}. For the term involving $\xi^x_k$, since $\frac{\lambda_x\mu_k}{1+\lambda_x\mu_k}\le 1$,
\[
\bigg\vert n^{-1}\sum_k \xi^x_kc_k\frac{\lambda_x\mu_k}{1+\lambda_x\mu_k}\frac{\lambda\mu_k}{1+\lambda\mu_k}\bigg\vert\le\bigg\vert n^{-1}\sum_k \xi^x_kc_k\frac{\lambda\mu_k}{1+\lambda\mu_k}\bigg\vert=\lo(n^{-1/2})
\]
by the proof of Lemma \ref{lem:tech_rice} (c). This proves (c).


For (d) we have that
\[
n^{-1}\xb^T(\IB-\SBlx)(\IB-\SBl)\SBlx\xb=n^{-1}\sum_k(c^x_k+\xi^x_k)^2\frac{\lambda_x\mu_k}{(1+\lambda_x\mu_k)^2}\frac{\lambda\mu_k}{1+\lambda\mu_k}.
\]
For the term involving $(c^x_k)^2$, Cauchy-Schwarz implies that
\begin{eqnarray*}
n^{-1}\sum_k(c^x_k)^2\frac{\lambda_x\mu_k}{(1+\lambda_x\mu_k)^2}\frac{\lambda\mu_k}{1+\lambda\mu_k}&\le&n^{-1}\sum_k(c^x_k)^2\frac{\lambda_x\mu_k}{1+\lambda_x\mu_k}\frac{\lambda\mu_k}{1+\lambda\mu_k}\\
&\le& \big(\Btp(f^x,\lambda_x)\Btp(f^x,\lambda)\big)^{1/2}=\LO\big((\lambda\lambda_x)^{1/2})
\end{eqnarray*}
by Lemma \ref{lem:sqbias}.
For the term involving $(\xi^x_k)^2$ we use (A3) and Lemma \ref{lem:trace} to see that
\begin{eqnarray*}
n^{-1}\sum_k(\xi^x_k)^2\frac{\lambda_x\mu_k}{(1+\lambda_x\mu_k)^2}\frac{\lambda\mu_k}{1+\lambda\mu_k}&\le&\sup_k(\xi^x_k)^2 n^{-1}\sum_k\frac{1}{1+\lambda_x\mu_k}\\
&=&\LO\big((\log^2n)n^{-1/2-\tau}\big)=\lo(n^{-1/2}).
\end{eqnarray*}
This proves (d)

For (e) we have that
\[
n^{-1}\xb^T\SBlx(\IB-\SBl)^2\SBlx\xb=n^{-1}\sum_k(c^x_k+\xi^x_k)^2\frac{1}{(1+\lambda_x\mu_k)^2}\bigg(\frac{\lambda\mu_k}{1+\lambda\mu_k}\bigg)^2.
\]
For the term involving $(c^x_k)^2$ we see that
\begin{eqnarray*}
n^{-1}\sum_k(c^x_k)^2\frac{1}{(1+\lambda_x\mu_k)^2}\bigg(\frac{\lambda\mu_k}{1+\lambda\mu_k}\bigg)^2&\le&n^{-1}\sum_k(c^x_k)^2\bigg(\frac{\lambda\mu_k}{1+\lambda\mu_k}\bigg)^2\\
&=&\Btp(f^x,\lambda)=\LO(\lambda).
\end{eqnarray*}
For the term involving $(\xi^x_k)^2$
\begin{eqnarray*}
n^{-1}\sum_k(\xi^x_k)^2\frac{1}{(1+\lambda_x\mu_k)^2}\bigg(\frac{\lambda\mu_k}{1+\lambda\mu_k}\bigg)^2&\le&n^{-1}\sup_k(\xi^x_k)^2\sum_k\frac{1}{(1+\lambda_x\mu_k)^2}\\
&=&\LO\big(n^{-1}(\log^2n)\lambda_x^{-d/2m}\big)
\end{eqnarray*}
by assumption (A3) and Lemma \ref{lem:trace}. This proves (e).

%

For (f) we write
\begin{align}\label{eqn:tech_lem_g}
n^{-1}\xb^T\big[\SBl+(\IB-\SBl)\SBlx]^T[\SBl+(\IB-\SBl)\SBlx\big]\xb\qquad\qquad\qquad\\\nonumber
=n^{-1}\big(\xb^T\SBl^2\xb+2\xb^T\SBl(\IB-\SBl)\SBlx\xb+\xb^T\SBlx(\IB-\SBl)^2\SBlx\big).
\end{align}
For the first term in \eqref{eqn:tech_lem_g}, $n^{-1}\xb^T\SBl^2\xb=\LO(1)$ by Lemma \ref{lem:tech_rice} (d). 
For the second term in \eqref{eqn:tech_lem_g} we see that
\begin{eqnarray*}
n^{-1}\xb^T\SBl(\IB-\SBl)\SBlx\xb&=&n^{-1}\sum_k(c^x_k+\xi^x_k)^2\frac{1}{1+\lambda_x\mu_k}\frac{\lambda\mu_k}{(1+\lambda\mu_k)^2}\\
&\le&n^{-1}\sum_k(c^x_k+\xi^x_k)^2\frac{\lambda\mu_k}{1+\lambda\mu_k}=n^{-1}\xb^T(\IB-\SBl)\xb=\LO(1)
\end{eqnarray*}
by Lemma \ref{lem:tech_rice} (a).
From (e), the third term in \eqref{eqn:tech_lem_g} is given by
\begin{eqnarray*}
n^{-1}\xb^T\SBlx(\IB-\SBl)^2\SBlx&=&\LO(\lambda)+\LO(n^{-1}\lambda_x^{-d/2m}\log^2n)\\
&\approx& \LO(n^{-\delta})+\LO(n^{-(1-\delta_x d/2m)}\log^2n)=\LO(1).
\end{eqnarray*}
This proves (f).
\section{Proofs of main results}\label{app:proofs_main}
In this appendix we prove the results given in Sections $\secasymps$ and $\secasymts$ of the paper.

\subsection*{Proof of Theorem \thmpsest}
Let $\fb=(f(\tb_1),\ldots,f(\tb_n))^T$. Since $\E(\yb)=\beta\xb+\fb$, the expression (\eqnpsestimates) in the paper shows that
\begin{eqnarray*}
\E(\betahps)-\beta&=&\left(\xb^T(\IB-\SBl)\xb\right)^{-1}\xb^T(\IB-\SBl)(\beta\xb+\fb)-\beta\\
&=&\big(n^{-1}\xb^T(\IB-\SBl)\xb\big)^{-1}\big(n^{-1}\xb^T(\IB-\SBl)\fb\big)\\
&=&\lo(n^{-1/2})+\LO(\lambda^{1/2})
\end{eqnarray*}
by Lemma \ref{lem:tech_rice} (a) and (c).

Similarly, since $\var(\yb)=\sigma^2\IB$, (\eqnpsestimates) in the paper shows that 
\begin{eqnarray*}
n\var(\betahps)&=&n\sigma^2 \big(\xb^T(\IB-\SBl)\xb\big)^{-1}\xb^T(\IB-\SBl)^2\xb\big(\xb^T(\IB-\SBl)\xb\big)^{-1}\\
&=&\sigma^2\big(n^{-1}\xb^T(\IB-\SBl)\xb\big)^{-1}\big(n^{-1}\xb^T(\IB-\SBl)^2\xb\big)\big(n^{-1}\xb^T(\IB-\SBl)\xb\big)^{-1}\\
&\rightarrow& \sigma^2/\sigma_x^2 \quad\text{as }n\rightarrow \infty
\end{eqnarray*}
by Lemma \ref{lem:tech_rice} (a) and (b). 

\subsection*{Proof of Theorem \thmpslambda}
Let $\fb=(f(\tb_1),\ldots,f(\tb_n))^T$ and $\eepsilon=(\epsilon_1,\ldots,\epsilon_n)^T$ so that $\yb=\beta\xb+\fb+\eepsilon$. 

Since $\fbhps=\SBl(\yb-\betahps\xb)$ by (\eqnpsestimates) in the paper,
\[
\E(\fbhps)-\fb=-\big(\E(\betahps)-\beta\big)\SBl\xb-(\IB-\SBl)\fb.
\]
We therefore see that
\begin{eqnarray*}
\Bps(f,\lambda)&=&n^{-1}\Vert\E(\fbhps)-\fb\Vert^2\\
&\le&n^{-1}\big\Vert\big(\E(\betahps)-\beta\big)\SBl\xb\big\Vert^2+n^{-1}\Vert(\IB-\SBl)\fb\Vert^2\\
&=&\big(\E(\betahps)-\beta\big)^2n^{-1}\xb^T\SBl^2\xb+\Btp(f,\lambda)\\
&=&\big(\lo(n^{-1})+\LO(\lambda)\big)\LO(1)+\LO(\lambda)=\LO(\lambda)
\end{eqnarray*}
by Theorem \thmpsest (a), Lemma \ref{lem:tech_rice} (d) and Lemma \ref{lem:sqbias}. This proves part (a).

For (b), firstly note that
\[
\fbhps-\E(\fbhps)=\SBl\eepsilon-\big(\betahps-\E(\betahps)\big)\SBl\xb.
\]
We therefore see that
\begin{eqnarray*}
\Vps(f,\lambda)&=&n^{-1}\E\big(\Vert \fbhps-\E(\fbhps)\Vert^2\big)\\
&\le&n^{-1}\E(\eepsilon^T\SBl^2\eepsilon)+\E\big[(\betahps-\E(\betahps))^2\big]n^{-1}\xb^T\SBl^2\xb\\
&=&n^{-1}\sigma^2\mathrm{Tr}(\SBl^2)+\var(\betahps)\LO(1)\\
&=&\LO(n^{-1}\lambda^{-d/2m})+\LO(n^{-1})=\LO(n^{-1}\lambda^{-d/2m})
\end{eqnarray*}
by Lemma \ref{lem:trace}, Theorem \thmpsest (b) and Lemma \ref{lem:tech_rice} (d). This proves part (b).

Finally, recall that
\[
\amse(\fbhps)=\Bps(f,\lambda)+\Vps(f,\lambda).
\]
From the above, we see that the bias term increases with $\lambda$ while the variance term decreases with $\lambda$ so that the optimal rate for minimising $\amse(\fbhps)$ is achieved when $\LO(\lambda)=\LO(n^{-1}\lambda^{-d/2m})$. This leads to an optimal rate of $\lambda=\LO(n^{-2m/(2m+d)})$. At this rate for $\lambda$, $\Bps(f,\lambda)$ and $\Vps(f,\lambda)$ converge at the same rate of $\LO(n^{-2m/(2m+d)})$.

\subsection*{Proof of Theorem \thmtsest}
Let 
\begin{eqnarray*}
\bb&=&(\IB-\SBl)(\IB-\SBlx)\xb\\
a_1&=&n^{-1}\bb^T(\IB-\SBlx)\xb=n^{-1}\xb^T(\IB-\SBlx)(\IB-\SBl)(\IB-\SBlx)\xb\\
a_2&=&n^{-1}\bb^T\bb=n^{-1}\xb^T(\IB-\SBlx)(\IB-\SBl)^2(\IB-\SBlx)\xb.
\end{eqnarray*}
By Lemma \ref{lem:tech} (a) and (b), $a_1\rightarrow \sigma_x^2$ and $a_2\rightarrow \sigma_x^2$ as $n\rightarrow \infty$. From (\eqntsestimatesb) in the paper we see that
\[
\betahts=(na_1)^{-1}\bb^T\yb.
\]
Therefore, since $\E(\yb)=\beta\xb+\fb$ where $\fb=(f(\tb_1),\ldots,f(\tb_n))$,
\begin{eqnarray*}
\E(\betahts)-\beta&=&(na_1)^{-1}\big((\bb^T\xb-na_1)\beta+\bb^T\fb\big)\\
&=&(na_1)^{-1}\big(\bb^T\SBlx\xb\beta+\bb^T\fb\big)\\
&=&a_1^{-1}\big(n^{-1}\xb^T(\IB-\SBlx)(\IB-\SBl)\SBlx\xb\beta+n^{-1}\xb^T(\IB-\SBlx)(\IB-\SBl)\fb\big)\\
&=&\lo(n^{-1/2})+\LO((\lambda\lambda_x)^{1/2})
\end{eqnarray*}
by Lemma \ref{lem:tech} (d) and (c). This proves part (a).

For part (b), since $\var(\yb)=\sigma^2\IB$, we see that
\begin{eqnarray*}
n\var(\betahts)&=&n(na_1)^{-2}\bb^T(\sigma^2\IB)\bb\\
&=&(\sigma^2a_2)/a_1^2\\
&\rightarrow& \sigma^2/\sigma_x^2\quad \text{as }n\rightarrow\infty.
\end{eqnarray*}
This proves (b).

\subsection*{Proof of Theorem \thmtslambda}
Let $\fb=(f(\tb_1),\ldots,f(\tb_n))^T$ and $\eepsilon=(\epsilon_1,\ldots,\epsilon_n)^T$ so that $\yb=\beta\xb+\fb+\eepsilon$. 
Since by (\eqntsestimatesf) in the paper
\[
\fbhts=\SBl\yb-\big(\SBl+(\IB-\SBl)\SBlx\big)\betahts\xb,
\]
we have that
\[
\E(\fbhts)-\fb=-(\IB-\SBl)\fb-(\E(\betahts)-\beta)\big(\SBl+(\IB-\SBl)\SBlx\big)\xb-\beta(\IB-\SBl)\SBlx\xb.
\]
Since $n^{-1}\big\Vert(\IB-\SBl)\fb\big\Vert^2=\Btp(f,\lambda)$, we therefore see that
\begin{eqnarray*}
\Bts(f,\lambda,\lambda_x)&=&n^{-1}\Vert\E(\fbhts)-\fb\Vert^2\\
&\le&n^{-1}\big\Vert(\IB-\SBl)\fb\big\Vert^2+\big(\E(\betahts)-\beta\big)^2n^{-1}\big\Vert\big(\SBl+(\IB-\SBl)\SBlx\big)\xb\big\Vert^2\\
& &+ \beta^2n^{-1}\big\Vert(\IB-\SBl)\SBlx\xb\big\Vert^2\\
&=&\LO(\lambda)+\big(\lo(n^{-1})+\LO(\lambda\lambda_x)\big)\LO(1)+\LO(\lambda)+\LO(n^{-1}\lambda_x^{d/2m}\log^2n)\\
&=&\LO(\lambda)+\LO(n^{-1}\lambda_x^{-d/2m}\log^2 n)
\end{eqnarray*}
by Lemma \ref{lem:sqbias}, Theorem \thmtsest (a) and Lemma \ref{lem:tech_rice} (f) and (e). This proves part (a).

For (b), note that
\[
\fbhts-\E(\fbhts)=\SBl\eepsilon-\big(\betahts-\E(\betahts)\big)\big(\SBl+(\IB-\SBl)\SBlx\big)\xb.
\]
We therefore see that
\begin{eqnarray*}
\Vts(f,\lambda,\lambda_x)&=&n^{-1}\E\big[\Vert \fbhts-\E(\fbhts)\Vert^2\big]\\
&\le&n^{-1}\E[\eepsilon^T\SBl^2\eepsilon]-\E\big[(\betahts-\E(\betahts))^2\big]n^{-1}\big\Vert\big(\SBl+(\IB-\SBl)\SBlx\big)\xb\big\Vert^2\\
&=&n^{-1}\sigma^2\mathrm{Tr}(\SBl^2)+\var(\betahts)\LO(1)\\
&=&\LO(n^{-1}\lambda^{-d/2m})+\LO(n^{-1})=\LO(n^{-1}\lambda^{-d/2m})
\end{eqnarray*}
by Lemma \ref{lem:tech} (f), Lemma \ref{lem:trace} and Theorem \thmtsest (b). This proves part (b).

Finally, recall that
\[
\amse(\fbhts)=\Bts(f,\lambda,\lambda_x)+\Vts(f,\lambda,\lambda_x).
\]
From the above we see that the bias term increases with $\lambda$ while the variance term decreases with $\lambda$ so that the optimal rate for minimising $\amse(\fbhts)$ is achieved when $\LO(\lambda)=\LO(n^{-1}\lambda^{-d/2m})$. This leads to an optimal rate of $\lambda=\LO(n^{-2m/(2m+d)})$. Since we have assumed that the convergence rates for $\Bts(f,\lambda,\lambda_x)$ and $\Vts(f,\lambda,\lambda_x)$ are equal, the optimal rate for $\lambda_x$ is then obtained when $\LO(n^{-1}\lambda_x^{-d/2m}\log^2 n)=\LO(n^{-2m/(2m+d)})$ which leads to $\LO(\lambda_x)=n^{-2m/(2m+d)}(\log n)^{4m/d}$.

\subsection*{Proof of Corollary \corts}
Theorem \thmtsest (b) shows that we need $\E(\betahts)-\beta=\lo(n^{-1/2})$ to ensure that the bias converges faster than the standard deviation. Part (a) of the same theorem shows that this required rate can be achieved if $\lambda\lambda_x=\lo(n^{-1})$. Suppose $\lambda$ and $\lambda_x$ converge at their optimal rates from Theorem \thmtslambda. Then since for any $\epsilon>0$,
\[
n^{-2m/(2m+d)}(\log n)^{4m/d}=\lo(n^{-2m/(2m+d)+\epsilon}),
\]
we have that
\[
\lambda\lambda_x=\lo(n^{-4m/(2m+d)+\epsilon})=\lo(n^{-1})
\]
if we choose $\epsilon=\frac{2m-d}{2m+d}$. This proves the result.

\section{Partial residual estimates}\label{app:pr}
In this appendix we consider, as an aside, the asymptotic behaviour of the partial residual estimates introduced by Denby (1986) and, independently, by \citet{speckman1988}, which are the estimates we obtain using the gSEM approach of \citet{thaden2018}. Here we adapt the method used in Sections $\secasymps$ and $\secasymts$ of the paper for estimates in the spatial and spatial+ models to show how the asymptotic results in \citet{chen_shiau1991} for the partial residual estimates generalise from the one-dimensional model to dimensions $d\ge 1$. We show that, as is the case for the spatial+ model, the smoothing-induced bias in the covariate effect estimate goes to $0$ faster than the standard deviation, i.e.\ the partial residual estimates also avoid the problem of disproportionate smoothing-induced bias.

For a given value $\lambda>0$ of the smoothing parameter, the partial residual estimates for the covariate effect $\beta$ and the unknown smooth effect $\fb=(f(\tb_1),\ldots,f(\tb_n))^T$ in the model (\eqnspatial) of the paper, are defined as
\begin{eqnarray}\label{eqn:pr_estimates}
\betahpr&=&\big(\xb^T(\IB-\SBl)^2\xb\big)^{-1}\xb^T(\IB-\SBl)^2\yb,\\
\fbhpr&=&\SBl(\yb-\betahpr\xb)\nonumber
\end{eqnarray}
where $\SBl$ is the smoother matrix. A similar argument to that of Section $\sects$ of the paper shows that these estimates are the ones we would obtain in the gSEM if, for simplicity, we used the same smoothing parameter in all regressions. That is, the estimate $\betahpr$ is the same as the estimated effect in the linear model given by
\[
r^y_i=\beta r^x_i+\epsilon_i, \quad \epsilon_i\underset{\text{iid}}\sim N(0,\sigma^2)
\]
where $\rb^x=(\IB-\SBl)\xb$ and $\rb^y=(\IB-\SBl)\yb$ are the residuals after fitting a thin plate spline to $\xb$ and $\yb$, respectively.

Minor adjustments to the proofs of Theorems $\thmpsest$ and $\thmpslambda$ and Corollary $\corps$ for the spatial model estimates lead to the following results. These results show that the asymptotic behaviour of the estimates $\betahpr$ and $\fbhpr$ is the same as that of the corresponding spatial model estimates, except for the rate of convergence of the bias of the covariate effect estimate $\betahpr$. More specifically, $\E(\betahpr)-\beta=\lo(n^{-1/2})+\LO(\lambda)$, whereas $\E(\betahps)-\beta=\lo(n^{-1/2})+\LO(\lambda^{1/2})$ and this difference is enough to ensure that the bias converges faster than the standard deviation when $\lambda$ converges at the optimal rate (for minimising the AMSE of the estimated spatial effect).
\begin{theorem}\label{thm:pr_est}
Suppose $\lambda\approx n^{-\delta}$ for some $0<\delta<1$, $f, f^x \in H^m(\Omega)$ are bounded and $m\ge d$.  Then for the partial residual estimate of $\beta$ we have that
\begin{description}
\item[(a)] $\E(\betahpr)-\beta=\lo(n^{-1/2})+\LO(\lambda)$,
\item[(b)] $n\var(\betahpr)\rightarrow \sigma^2/\sigma_x^2$ as $n\rightarrow\infty$.
\end{description}
In particular, $\var(\betahpr)=\LO(n^{-1})$ and we need $\lambda=\lo(n^{-1/2})$ to ensure that the bias converges faster than the standard deviation of $\betahpr$.
\end{theorem}
\begin{proof}
Let $\fb=(f(\tb_1),\ldots,f(\tb_n))^T$. Since $\E(\yb)=\beta\xb+\fb$, the expression \eqref{eqn:pr_estimates} shows that
\begin{eqnarray*}
\E(\betahpr)-\beta&=&\big(\xb^T(\IB-\SBl)^2\xb\big)^{-1}\xb^T(\IB-\SBl)^2(\beta\xb+\fb)-\beta\\
&=&\big(n^{-1}\xb^T(\IB-\SBl)^2\xb\big)^{-1}\big(n^{-1}\xb^T(\IB-\SBl)^2\fb\big)\\
&=&\lo(n^{-1/2})+\LO(\lambda)
\end{eqnarray*}
by Lemma \ref{lem:tech_rice} (b) and Lemma \ref{lem:tech} (c).

Similarly, since $\var(\yb)=\sigma^2\IB$, \eqref{eqn:pr_estimates} shows that 
\begin{eqnarray*}
n\var(\betahpr)&=&n\sigma^2 \big(\xb^T(\IB-\SBl)^2\xb\big)^{-1}\xb^T(\IB-\SBl)^4\xb\big(\xb^T(\IB-\SBl)^2\xb\big)^{-1}\\
&=&\sigma^2\big(n^{-1}\xb^T(\IB-\SBl)^2\xb\big)^{-1}\big(n^{-1}\xb^T(\IB-\SBl)^4\xb\big)\big(n^{-1}\xb^T(\IB-\SBl)^2\xb\big)^{-1}\\
&\rightarrow& \sigma^2/\sigma_x^2 \quad\text{as }n\rightarrow \infty
\end{eqnarray*}
by Lemma \ref{lem:tech_rice} (b) and Lemma \ref{lem:tech} (b). 
\end{proof}

\begin{theorem}\label{thm:pr_lambda}
Suppose $\lambda\approx n^{-\delta}$ for some $0<\delta<1$, $f, f^x\in H^m(\Omega)$ are bounded and $m\ge d$. Then the average squared bias $\Bpr(f,\lambda)$ and average variance $\Vpr(f,\lambda)$ of the partial residual estimate of $f$ satisfy
\begin{description}
\item[(a)] $\Bpr(f,\lambda)=n^{-1}\sum_i (\E((\fbhpr)_i)-f(\tb_i))^2=\LO(\lambda)$,
\item[(b)] $\Vpr(f,\lambda)=n^{-1}\sum_i\var ((\fbhpr)_i)=\LO(n^{-1}\lambda^{-d/2m})$.
\end{description}
In particular, the optimal rate for $\lambda$ in terms of minimising $\amse(\fbhpr)$ is 
$\lambda=\LO(n^{-2m/(2m+d)})$, and when $\lambda$ converges at this optimal rate, $\amse(\fbhpr)=\LO(n^{-2m/(2m+d)})$.
\end{theorem}
\begin{proof}
Let $\fb=(f(\tb_1),\ldots,f(\tb_n))^T$ and $\eepsilon=(\epsilon_1,\ldots,\epsilon_n)^T$ so that $\yb=\beta\xb+\fb+\eepsilon$. 

By (\ref{eqn:pr_estimates}), $\fbhpr=\SBl(\yb-\betahpr\xb)$ has the same format as the corresponding partial thin plate spline estimate, and therefore, 
\[
\E(\fbhpr)-\fb=-\big(\E(\betahpr)-\beta\big)\SBl\xb-(\IB-\SBl)\fb
\]
and
\[
\fbhpr-\E(\fbhpr)=\SBl\eepsilon-\big(\betahpr-\E(\betahpr)\big)\SBl\xb.
\]
as in the proof of Theorem \thmpslambda. For the derivation of $\Bpr(f,\lambda)$ and $\Vpr(f,\lambda)$, we can therefore apply the same proof where the only adjustment needed is the rate of convergence of the bias $\E(\betahpr)-\beta$.
\begin{eqnarray*}
\Bpr(f,\lambda)&=&n^{-1}\Vert\E(\fbhpr)-\fb\Vert^2\\
&\le&n^{-1}\big\Vert\big(\E(\betahpr)-\beta\big)\SBl\xb\big\Vert^2+n^{-1}\Vert(\IB-\SBl)\fb\Vert^2\\
&=&\big(\E(\betahpr)-\beta\big)^2n^{-1}\xb^T\SBl^2\xb+\Btp(f,\lambda)\\
&=&\big(\lo(n^{-1})+\LO(\lambda^2)\big)\LO(1)+\LO(\lambda)=\LO(\lambda)
\end{eqnarray*}
by Theorem \ref{thm:pr_est} (a), Lemma \ref{lem:tech_rice} (d) and Lemma \ref{lem:sqbias}. This proves part (a).

\begin{eqnarray*}
\Vpr(f,\lambda)&=&n^{-1}\E\big(\Vert \fbhpr-\E(\fbhpr)\Vert^2\big)\\
&\le&n^{-1}\E(\eepsilon^T\SBl^2\eepsilon)+\E\big[\big(\betahps-\E(\betahpr)\big)^2\big]n^{-1}\xb^T\SBl^2\xb\\
&=&n^{-1}\sigma^2\mathrm{Tr}(\SBl^2)+\var(\betahpr)\LO(1)\\
&=&\LO(n^{-1}\lambda^{-d/2m})+\LO(n^{-1})=\LO(n^{-1}\lambda^{-d/2m})
\end{eqnarray*}
by Lemma \ref{lem:trace}, Theorem \ref{thm:pr_est} (b) and Lemma \ref{lem:tech_rice} (d). This proves part (b).

The same argument as we used for the partial thin plate spline estimate $\fbhps$ shows that the optimal rate of convergence for minimising $\amse(\fbhpr)$ is achieved when $\LO(\lambda)=\LO(n^{-1}\lambda^{-d/2m})$, which leads to $\lambda=\LO(n^{-2m/(2m+d)})$ and $\amse(\fbhpr)=\LO(n^{-2m/(2m+d)})$.
\end{proof}

\begin{corollary}\label{cor:pr}
Suppose $\lambda\approx n^{-\delta}$ for some $0<\delta<1$, $f,f^x\in H^m(\Omega)$ are bounded and $m\ge d$. If $\lambda$ converges at the optimal rate in terms of minimising $\amse(\fbhpr)$, then
\[
\lambda=\lo(n^{-1/2}).
\]
In particular, the optimal rate for $\lambda$ ensures that the bias of the partial residual estimate $\betahpr$ converges faster than the standard deviation of the estimate. 
\end{corollary}


\section{Additional derivations for simulation results}\label{app:sim_der}

\subsection{Estimated effects in the unsmoothed spatial+ and gSEM models}
If no smoothing penalty is applied, the the spatial model, the gSEM and the spatial+ model are essentially the same, i.e.\ they have the same fitted values and the same unbiased estimate for the covariate effect. The spatial model is in this case an ordinary linear model where the columns in the model matrix are the covariate $\xb$ and the spatial basis vectors $\BBs$. 

The unsmoothed spatial+ model is a reparametrisation of the spatial model where the column $\xb$ in the model matrix is replaced by the spatial residuals $\rb^x=\xb-\fbh^x$ (where $\fbh^x$ are the fitted values of a spatial thin plate spline fitted to $\xb$). This does not change the overall column space as the difference $\fbh^x$ lies in the column space of $\BBs$. By the data generation process,
\begin{eqnarray*}
\yb&=&\beta\xb+\fb+\eepsilon^y\\
&=&\beta\rb^x+\beta\fbh^x-\zb-\zb'+\eepsilon^y, 
\end{eqnarray*}
with $\beta\fbh^x-\zb-\zb'$ in the column space of $\BBs$ and, therefore, the true effect of $\rb^x$ is the same as that of $\xb$. In fact, since $\rb^x$ is orthogonal to the spatial basis vectors, the estimated effect $\betah$ in the spatial+ model (\eqnmodplus) of the paper (and therefore in the spatial model (\eqnmod) of the paper) is obtained as
\begin{eqnarray*}
\betah&=&\big((\rb^x)^T\rb^x\big)^{-1}(\rb^x)^T\yb\\
&=&\beta+\big((\rb^x)^T\rb^x\big)^{-1}(\rb^x)^T\eepsilon^y.
\end{eqnarray*}


Similarly, for the unsmoothed gSEM, since
\[
\rb^y=\yb-\fbh^y=\beta\rb^x+\beta\fbh^x-\zb-\zb'-\fbh^y+\eepsilon^y,
\]
with $\beta\fbh^x-\zb-\zb'-\fbh^y$ in the column space of $\BBs$, the estimated effect $\betah$ of $\rb^x$ in the gSEM model (\eqnmodgsemtwo) of the paper is the same, namely,
\begin{eqnarray*}
\betah&=&\big((\rb^x)^T\rb^x\big)^{-1}(\rb^x)^T\rb^y\\
&=&\beta+\big((\rb^x)^T\rb^x\big)^{-1}(\rb^x)^T\eepsilon^y.
\end{eqnarray*}

Note that, since $\rb^x$ and $\eepsilon^y$ are independent, $\E(\betah)=\beta$, i.e.\ the the estimated covariate effect is unbiased.

\subsection{Non-Gaussian version of RSR}
Recall that in the Gaussian version of RSR, correlation between the covariate and spatial effect estimates is eliminated by restricting the spatial effect to the orthogonal complement of $\xb$. In Section \secspatialplusglm, we saw that 
estimation in the generalised version of the spatial model (i.e.\ (\eqnspatialglm) in the paper) corresponds to that of a Gaussian model in which the model matrix has columns $\tilde{\xb}=\sqrt{\WB}\xb$ and $\sqrt{\WB}\BBs$ with $\WB$ the weights matrix at convergence of the PIRLS algorithm. We can therefore define the generalised RSR model to be the same as the generalised spatial model but with the spatial basis vectors $\BBs$ in the model matrix replaced by
\[
\BBst=(\IB-\xb(\xb^T\WB\xb)^{-1}\xb^T\WB)\BBs.
\]
Then, by construction, the generalised RSR model corresponds to a Gaussian model for which the columns $\tilde{\xb}=\sqrt{\WB}\xb$ and $\sqrt{\WB}\BBst$  are orthogonal:
\[
\tilde{\xb}^T\sqrt{\WB}\BBst=\xb^T\WB\BBst=\boldsymbol{0}.
\]

\end{document}